\documentclass{llncs}
\usepackage{enumerate}
\usepackage{hyperref}
\usepackage{amsmath, amssymb}
\usepackage{inlinenum}
\usepackage{bussproofs}
\usepackage{xspace}

\newcommand{\voc}{\tau}
\newcommand{\Tr}{{\bf t}}
\newcommand{\Fa}{{\bf f}}
\newcommand{\Un}{{\bf u}}

\newcommand{\rul}{\leftarrow}
\newcommand{\D}{D}
\newcommand{\defp}[1]{{\voc^d_{#1}}}
\newcommand{\openp}[1]{{\voc^o_{#1}}}
\newcommand{\la}{\rightarrow}
\newcommand{\inferencerule}[2]{\displaystyle\dfrac {#1}{#2}}
\newcommand{\defin}[1]{\left\{ \begin{array}{l}#1\end{array} \right\}}
\newcommand{\tf}[2]{{#1}^{#2}}

\newcommand{\ignore}[1]{}

\newcommand{\res}[2]{{{#1}\rvert_{#2}}}
\newcommand{\ps}{{\bf LPC(ID)}\xspace}
\newcommand{\gs}{{\bf LK}\xspace}
\newcommand{\minisatid}{M{\sc ini}S{\sc at}(ID)\xspace}

\begin{document}
\title{\ps: A Sequent Calculus Proof System for Propositional Logic Extended with Inductive Definitions}
\author{Ping Hou\inst{1,2} \and Johan Wittocx\inst{2} \and Marc Denecker\inst{2} \\
 \institute {Computer Science Department, Carnegie Mellon University \and
             Department of Computer Science, Katholieke Universiteit Leuven,
   Belgium}}
\maketitle
\begin{abstract}
  The logic FO(ID) uses ideas from the field of logic programming to
  extend first order logic with non-monotone inductive
  definitions. Such logic formally extends logic programming,
  abductive logic programming and datalog, and thus formalizes the
  view on these formalisms as logics of (generalized) inductive
  definitions.  The goal of this paper is to study a deductive
  inference method for PC(ID), which is the propositional fragment of
  FO(ID). We introduce a formal proof system based on the sequent
  calculus (Gentzen-style deductive system) for this logic.  As PC(ID)
  is an integration of classical propositional logic and propositional
  inductive definitions, our sequent calculus proof system integrates
  inference rules for propositional calculus and definitions. We
  present the soundness and completeness of this proof system with
  respect to a slightly restricted fragment of
  PC(ID).
  We also provide some complexity results for PC(ID). By developing
  the proof system for PC(ID), it helps us to enhance the
  understanding of proof-theoretic foundations of FO(ID), and
  therefore to investigate useful proof systems for FO(ID).
\end{abstract}

\section{Introduction}\label{sec:intro}

In this paper, we study deductive methods for the propositional
fragment of FO(ID) \cite{DeneckerT08}. To motivate
this study, we need to say a few words about the origin and the
motivation of FO(ID).

Perhaps the two most important knowledge representation paradigms of
the moment are on the one hand, classical logic-based approaches such
as description logics \cite{DL2002}, and on the other hand, rule-based
approaches based on logic programming and extensions such as Answer
Set Programming and Abductive Logic Programming
\cite{Baral:book2003,Kakas93a}. The latter disciplines are rooted
firmly in the discipline of Non-Monotonic Reasoning
\cite{McCarthy86}. FO(ID) integrates both paradigms in a tight,
conceptually clean manner. The key to integrate ``rules'' into
classical logic (FO) is the observation that natural language, or more
precisely, the informal language of mathematicians, has an informal
rule-based construct: the construct of {\em inductive/recursive definitions}
(IDs).

\newcommand{\valI}{I}
\begin{figure}\begin{definition}
The transitive closure $T_G$ of a directed graph $G$ is defined by induction:
\begin{itemize}
\item $(x,y) \in T_G$ if $(x,y) \in G$;
\item $(x,y) \in T_G$ if for some vertex $z$, $(x,z)\in T_G$ and
  $(z,y)\in T_G$.
\end{itemize}
\end{definition}
\caption{Definition of Transitive closure\label{fig-trans}}
\end{figure}
\begin{figure}\begin{definition}The satisfaction relation $\models$ between $\sigma$-interpretations $I$ and propositional formulas over $\sigma$ is defined by structural induction:
\begin{itemize}
\item [ - ] $\valI \models p$ if $p$ is an atom and $p\in \valI$,
\item [ - ]  $\valI \models \psi \land \phi$ if $\valI\models \psi$ and
  $\valI\models \phi$,
\item [ - ]  $\valI \models \psi \lor \phi$ if $\valI\models \psi$ or
  $\valI\models \phi$,
\item  [ - ] $\valI \models \neg \psi$ if $\valI \not\models \psi$.
\end{itemize}
\end{definition}
\caption{Definition of satisfaction\label{fig-sat}}
\end{figure}

In Figure~\ref{fig-trans} and Figure~\ref{fig-sat}, we displayed two
prototypical examples of the most common forms of inductive
definitions in mathematics: monotone ones, respectively definitions by
induction over a well-founded order. As seen in these figures, both
are frequently represented as a set of informal rules.  These two
forms of inductive definitions are generalized by the concept of {\em
  iterated} inductive definitions (IID) \cite{Feferman81}. Inductive
definitions define their concept by describing how to \emph{construct}
it through a process of iterated application of rules starting from
the empty set. Definitions by induction over a well-founded order are
frequently non-monotone, as illustrated by the non-monotone rule
``$\valI\models\neg\psi$ if $\valI\not\models\psi$'' which derives the
satisfaction of $\neg\psi$ given the non-satisfaction of $\psi$.

Of course,  a definition is not just a set of material
implications. Thus, a sensible scientific research question is to
design a uniform, rule-based formalism for representing these forms of
definitions. Such a study is not only useful as a formal logic study
of the concept of inductive definition but it contributes to
the understanding of rule-based systems and thus, to the study of the
(formal and informal) semantics of logic programming and the
integration of classical logic-based and rule-based approaches to
knowledge representation.

Iterated inductive definitions have been studied in mathematical logic
\cite{Feferman81} but the formalisms there are not rule-based and
require an extremely tedious encoding of rules and well-founded
orderings into one complex formula \cite{DeneckerT08}. In several
papers \cite{Denecker98c,Denecker2001:TOCL,DeneckerT08}, it was argued
that, although unintended by its inventors, the rule-based formalism
of logic programming under the well-founded semantics
\cite{VanGelder91} and its extension to rules with FO-bodies in
\cite{VanGelder93} correctly formalizes the above mentioned forms of
inductive definitions. Stated differently, if we express an informal
inductive definition of one of the above kinds into a set of formal
rules $$\forall \bar{x} (P(\bar{t}) \rul \phi)$$ then the informal
semantics of the original definition matches the well-founded
semantics of the formal rule set. E.g., in a well-founded model of the
following ``literal'' translation of the definition in
Figure~\ref{fig-trans}:
\[ \defin{\forall x, y \ (T_G(x,y)  \rul G(x,y))\\
  \forall x, y \ (T_G(x,z) \rul (\exists z \ T_G(x,y) \land T_G(y,
  z)))}\] $T_G$ is interpreted as the transitive closure of the graph
interpreting $G$. A similar claim holds for the literal translation of
the definition of $\models$ in Figure~\ref{fig-sat}. Thus, the rule
formalism under the well-founded semantics provides the desired
uniform syntax and semantics for representing the above mentioned
forms of inductive definition construct.

There are several good arguments to integrate the above inductive
definition construct (and hence, this generalized form of logic
programming under the well-founded semantics) into FO. (1) FO and
definitions are complementary KR languages: FO is a base language very
suitable for expressing propositions, assertions or constraints while
it is well-known that, in general, inductive definitions cannot be
expressed in FO \cite{Libkin}. (2) Definitions are important for
KR. In the case of {\em non-inductive} definitions, their use for
defining terminology was argued long time ago in Brachman and
Levesque's seminal paper~\cite{Brachman82} and was the motivation for
developing description logics \cite{DL2002}. As for {\em inductive}
definitions, they are quite likely as important to declarative
Knowledge Representation as recursive functions and procedures are to
programming.  Applications of inductive definitions abound in KR:
various instances of transitive closure, definitions of recursive
types and of concepts defined over recursive types, descriptions of
dynamic worlds through definitions of states in terms of past states
and effects of actions, etc.  In~\cite{Denecker/Ternovska:2007:AI}, a
formalization of situation calculus in terms of iterated inductive
definitions in FO(ID) yields an elegant and very general solution for
the ramification problem in the context of the situation calculus. (3)
Inductive definitions are also an interesting {\em Non-Monotonic
  Reasoning} language construct. A logic is non-monotonic if adding
new expressions to a theory may invalidate previous
inferences. Obviously, adding a new rule to an inductive definition
defines a different set and hence, this operation may invalidate
previous inferences\footnote{Observe that the concept of
  (non-)monotonicity is used here in two different ways. Adding a rule
  to a {\em monotone inductive definition} is a {\em non-monotonic
    reasoning} operation.}.  
One of the main non-monotonic reasoning principles is the Closed World
Assumption (CWA) \cite{Reiter78a}. The intuition underlying CWA is
that ``an atom is false unless it can be proven''.  This matches with
an inductive definition in which a defined atom $P(\bar{t})$ is false
unless it is explicitly derived by one of its rules $P(\bar{t}) \rul
\psi$ during the construction process.  Hence, inductive definitions can be viewed as a
very precise and well-understood form of Closed World
Assumption. Moreover, it is well-known that rule formalisms under CWA
can be used to represent many useful forms of {\em defaults}. The
correspondence between CWA and inductive definition construct implies that the methodologies to
represent defaults developed in, e.g., logic programming, can be used
in an inductive definition formalism as well.  Domain Closure
\cite{McCarthy77} is another important non-monotonic reasoning
principle that can be expressed through inductive definitions
\cite{DeneckerT08}.

All the above provides a strong motivation for adding inductive definitions to
FO. Thus, the resulting logic FO(ID) extends FO not only with an inductive definition
construct but also with an expressive and precise non-monotonic
reasoning principle.  Not surprisingly, the logic FO(ID) is strongly
tied to many other logics. It is an extension of FO with inductive definitions and a
conceptually clean integration of FO and LP. It integrates monotonic
and non-monotonic logics. The inductive definition construct of FO(ID) formally
generalizes Datalog \cite{Abiteboul95}: this is a natural match, given
that Datalog programs aim to \emph{define} queries and views. FO(ID)
is also strongly related to fixpoint logics.  Monotone definitions in
FO(ID) are a different -rule-based- syntactic sugar of the fixpoint
formulas of Least Fixpoint Logic (LFP)
\cite{MI/Park70,TCS/Park76}. Last but not least, FO(ID), being a
conceptually clean, well-founded integration of rules into classical
logic, might play a unifying role in the current attempts of extending
FO-based description logics with rules \cite{ESWC/VennekensD09}. It
thus appears that FO(ID) occupies quite a central position in the spectrum
of computational and knowledge representation logics.


Several attempts to build inference systems for FO(ID) are underway.
One line of research is the development of finite model generators~\cite{lash06/MarienWD06,lpar05/MarienMDB05,sat/MarienWDB08,lash08/MarienWD08}
. They have similar applications and speed as current
Answer Set Programming solvers \cite{lpnmr/GebserLNNST07,lpnmr/ASPcompetition09}. However, in
this paper we study a more traditional form of inference: deduction.
As for every formal logical system, the development of deductive
inference methods for FO(ID) is an important research topic.  There is
no hope of course to build a complete proof system of FO(ID). Indeed,
inductive definability leads to undecidability, not even
semi-decidability. As such, the task we set out for this paper is
restricted to the development of a sound proof system and a
decidable fragment of FO(ID).


The goal of this paper is to extend the propositional part of
Gentzen's sequent calculus to obtain a proof system for PC(ID), the
propositional fragment of FO(ID). We view our work as an initial
investigation to build proof systems for (fragments of) FO(ID). In
proof theory, Gentzen's sequent calculus \gs~\cite{Gentzen35,Szabo69}
is a widely known proof system for first order logic.  The sequent
calculus is well-suited  to a goal-directed approach for constructing
logical derivations.
The advantage of the method
is its theoretical elegance and the fact that it focuses the proof search,
with applicable proof rules constrained by logical connectives
appearing in the current goal sequent. Our work is inspired by the one
of Compton, who used sequent calculus (Gentzen-style deductive system)
methods in~\cite{Compton93,Compton94} to investigate sound and
complete deductive inference methods for existential least fixpoint
logic and stratified least fixpoint logic. Existential least fixpoint
logic, as described in~\cite{Compton93}, is a logic with a least
fixpoint operator but only existential quantification and stratified
least fixpoint logic, as shown in~\cite{Compton94}, is a logic with a
least fixpoint operator and characterizes the expressibility of
stratified logic programs. Indeed, these two logics without nested
least fixpoint expressions can be viewed as fragments of FO(ID).

The contributions of this paper can be summarized as follows:
\begin{enumerate}
 \item \label{1contribution} We introduce a sequent calculus for PC(ID).
 \item \label{2contribution} We prove that the deductive system is sound and complete for a slightly restricted fragment of
PC(ID).
 \item \label{4contribution} We provide some complexity results for PC(ID).
\end{enumerate}

By developing a proof system for PC(ID), we want set a first step to
enhance the understanding of proof-theoretic foundations of
FO(ID). One application of this work could be for the development of
tools to check the correctness of the outputs generated by PC(ID)
model generators such as \minisatid~\cite{sat/MarienWDB08}. Given a PC(ID) theory $T$ as input, such a solver outputs a model for $T$ or concludes that $T$ is unsatisfiable. In the former
case, an independent {\em model checker} can be used to check whether
the output is indeed a model of $T$. However, when the solver
concludes that $T$ is unsatisfiable, it is less obvious how to check
the correctness of this answer. One solution is to transform a trace
of the solvers computation into a proof of unsatisfiability in a
PC(ID) proof system. An independent {\em proof checker} can then be
used to check this formal proof. Model and proof checkers can be a
great help to detect bugs in model generators. An analogous checker
for the Boolean Satisfiability problem (SAT) solvers was described in
\cite{LinTaoZhang03}.

On the longer run, we view our work also as a first step towards the
development of proof systems and decidable fragments of FO(ID). A
potential use of this is in the field of description logics.
Deductive reasoning is the distinguished form of inference of
Description Logics. Given the efforts to extend Description Logics
with rules and the fact that FO(ID) offers a natural, clean
integration of a very useful form of rules in FO, it seems that
research on decidable fragments of FO(ID) could play a useful role
in that area.


The structure of this paper is as follows. We introduce PC(ID) in
Section~\ref{sec:prel}. We present a sequent calculus proof system for
PC(ID) in Section~\ref{sec:dedsys}. The main results of the soundness
and completeness of the proof system are investigated in
Section~\ref{sec:result}. We provide some complexity results for
PC(ID) in Section~\ref{sec:complexity}. We finish with conclusions,
related and future work.

\section{Preliminaries}\label{sec:prel}

In this section, we present PC(ID), which is the propositional fragment of FO(ID)~\cite{DeneckerT08}. Observe that PC(ID) is an extension of propositional calculus (PC) with propositional inductive definitions (IDs).

\subsection{Syntax}
A propositional vocabulary $\voc$ is a set of propositional atoms. A formula of propositional calculus over $\voc$, or briefly, a PC-formula over $\voc$, is inductively defined as:
\begin{itemize}
\item an atom in $\voc$ is a PC-formula over $\voc$;
\item if $F$ is a PC-formula over $\voc$, then so is $\neg F$;
\item if $F_1, F_2$ are PC-formulas over $\voc$, then so are $F_1 \land F_2$ and $ F_1 \lor F_2$.
\end{itemize}

We use the following standard abbreviations: $F_1 \supset F_2$ for $\neg F_1 \lor F_2$ and $F_1 \equiv F_2$ for $(F_1 \land F_2) \lor (\neg F_1 \land \neg F_2)$. A {\em literal} is an atom $P$ or its negation $\neg P$. An atom $P$ has a {\em negative (positive) occurrence} in formula $F$ if $P$ has an occurrence in the scope of an odd (even) number of occurrences of the negation symbol $\neg$ in $F$.

A \emph{definition} $\D$ over $\voc$ is a finite set of rules of the form:
\[P \rul \varphi,\] where $P\in\voc$ and $\varphi$ is a PC-formula
over $\voc$. Note that the symbol ``$\rul$'' is a new symbol, which
must be distinguished from (the inverse of) material implication
$\supset$. For a rule of the above form, the atom $P$ is called the
{\em head} of the rule while $\varphi$ is known as its {\em body}. An
atom appearing in the head of a rule of $\D$ is called a
\emph{defined} atom of $\D$, any other atom is called an \emph{open}
atom of $\D$. We denote the set of defined atoms by $\defp{D}$ and
that of all open ones by $\openp{D}$.  We call a definition $\D$
{\em positive} if its defined symbols have only positive occurrences in rule
bodies (i.e., occur in the scope of an even number of negation
symbols).

$\D$ is called {\em inductive} or {\em recursive} in predicate $P$ if
its dependency relation $\prec$ satisfies $P\prec P$. Here, the
dependency relation $\prec$ of $\D$ on $\voc$ is the transitive
closure of the set of all pairs $(Q,P)$ such that for some rule $P
\rul \varphi\in\D$, $Q$ occurs in $\varphi$. The intended {\em
  informal semantics} of a formal definition $\D$ is given by
understanding it as a -possibly inductive- definition of the defined
symbols in terms of the open symbols. This understanding is clear in
case of positive definitions and the corresponding formal semantics is
obvious. In the next sections, we consider how this view extends to
arbitrary non-positive definitions.

A PC(ID)-formula over $\voc$ is 
defined by the following induction:
\begin{itemize}
\item an atom in $\voc$ is a PC(ID)-formula over $\voc$;
\item a definition over $\voc$ is a PC(ID)-formula over $\voc$;
\item if $F$ is a PC(ID)-formula over $\voc$, then so is $\neg F$;
\item if $F_1, F_2$ are PC(ID)-formulas over $\voc$, then so are $F_1 \land F_2$ and $ F_1 \lor F_2$.
\end{itemize}

A PC(ID) theory over $\voc$ is a set of PC(ID)-formulas over $\voc$.

Any definition containing multiple rules with the same atom in the
head can be easily transformed into a definition with only one rule
per defined atom. We illustrate this by the following example.

\begin{example}\label{ex:mu}
The following definition
\[ \defin{P \rul O_1 \land Q\\
          P \rul P\\
	  Q \rul Q \land P\\
	  Q \rul O_2} \]
is equivalent to this one:
\[ \defin{ P \rul (O_1 \land Q) \lor P\\
           Q \rul (Q \land P) \lor O_2}.\]
\end{example}

As we mentioned in Section~\ref{sec:intro}, monotone definitions in
FO(ID) are a different -rule-based- syntactic sugar of the fixpoint
formulas of Least Fixpoint Logic (LFP). We now illustrate the relation
between a propositional inductive definition and a propositional
least fixpoint expression in fixpoint logics.

A {\em propositional least fixpoint expression} is of the form:
 $$[{LFP}_{P_1, \ldots, P_n}(\theta_1, \ldots, \theta_n)]\psi,$$
where for each $i \in [1, \ldots, n]$, $P_i$ is a propositional atom, $\theta_i$ is either a propositional formula or a propositional least fixpoint expression, $\psi$ is either a propositional formula or a propositional least fixpoint expression, and $P_i$ occurs only positively in $\theta_i$ and $\psi$.
Note that the subformulas $\psi, \theta_1, \ldots, \theta_n$ of a
least fixpoint expression
$[{LFP}_{P_1, \ldots, P_n}(\theta_1, \ldots, \theta_n)]\psi$ may contain least fixpoint expressions. Indeed, nesting of least fixpoint expressions is allowed in fixpoint logics. But nesting of definitions is not allowed in PC(ID). All subformulas $\psi, \theta_1, \ldots, \theta_n$ of an unnested least fixpoint expression contain only positive occurrences of each atom $P_i$. It is worth mentioning that the unnested least fixpoint expression $[{LFP}_{P_1, \ldots, P_n}(\theta_1, \ldots, \theta_n)]\psi$, where $\theta_1, \ldots, \theta_n, \psi$ may not contain least fixpoint expressions, corresponds exactly to the second order PC(ID)-formula
$$\exists P_1 \dots P_n \left(\defin{P_1 \rul \theta_1\\
    \vdots\\
    P_n\rul \theta_n} \land \psi\right).$$ However, such a
correspondence does not hold for nested least fixpoint expressions
since only PC-formulas are allowed as bodies of rules in definitions.

In summary, the differences between the definition construct and the
fixpoint definitions are:
\begin{itemize}
\item The fixpoint notation is formula-based and defines predicate
  variables with scope restricted to the fixpoint expression while a
  definition construct is rule-based and defines predicate
  constants. (These are ``syntactic sugar'' differences.)
\item Fixpoint expressions can be nested while definitions cannot. On
  the other hand, in fixpoint expressions, the defined variables can
  occur only positively in the defining formulas, while in
  definitions, the defined predicates can occur negatively in rule
  bodies.
\end{itemize}
The relation between definitions and LFP are investigated in
\cite{pHou09ASP}.

\subsection{Semantics} \label{sec:sem}

In this section, we formalize the informal semantics of the two most
common forms of inductive definition, monotone inductive definitions
(e.g., the definition of transitive closure, Figure~\ref{fig-trans})
and definitions over a well-founded order (e.g., the definition of the
satisfaction relation $\models$, Figure~\ref{fig-sat}), and their
generalization, the notion of an iterated inductive definition. These
informal types of definitions might  be roughly characterized as follows:
\begin{itemize}
\item The rules of a monotone inductive definition of a set add
  objects to the defined set given the {\em presence} of certain other
  objects in the set.
\item For an inductive definition over some (strict)
  well-founded order, a rule adds an object $x$ given the {\em
    presence} or {\em absence} of certain other {\em strictly smaller}
  objects  in the set.
\item Finally, an iterated inductive definition is associated with a
  well-founded semi-order  \footnote{A semi-order $\leq$ is a
    transitive reflexive binary relation. Two elements $x, y$ are
    $\leq$-equivalent if $x\leq y$ and $y\leq x$, and $x$ is strictly
    less than $y$ if $x \leq y$ and $y \not \leq x$ are not
    equivalent. A semi-order is well-founded if it has no infinite
    strictly descending chains $x_0 > x_1 > x_2 > \dots$.}  such that
  each rule adds an object $x$ given the {\em presence} of some other
  {\em less or equivalent} objects in the defined set and the {\em
    absence} of some {\em strictly less} objects.
\end{itemize} According to this characterizations, iterated inductive
definitions generalize the other types. Non-monotonicity of the two
latter types of definitions stem from rule conditions that refer to
the absence of objects in the defined set (as in the condition of
``$\valI\models\neg\psi$ if $\valI \not\models \psi$'').  Adding a new
element to the set might violate a condition what was previously
satisfied. For an extensive argument how the well-founded semantics
uniformally formalizes these three principles, we refer to
\cite{Denecker2001:TOCL,DeneckerT08}. Below, we just sketch the main
intuitions.

As we all know, the set defined by any of the aforementioned forms of
inductive definitions can be obtained constructively as the limit of
an increasing sequence of sets, by starting with the empty set and
iteratively applying unsatisfied rules until saturation.  A key
difference between  monotone
definitions and non-monotone inductive definitions is that in the
first, once the condition of a rule is satisfied in some intermediate
set, it holds in all later stages of the construction. This is not the
case for non-monotone inductive definitions. E.g., in the construction
of $\models$, the set of formulas $\psi$ for which the condition of
the rule ``$\valI \models \neg \psi$ {\em if } $\valI \not\models
\psi$'' holds, initially contains all formulas and gradually
decreases. As a consequence, the order of rule applications is
arbitrary for monotone inductive definitions but matters for
non-monotone definitions. There, it is critical to delay application
of an unsatisfied rule until it is certain that its condition will not
be falsified by later rule applications. This is taken care of by
applying the rules along the well-founded order provided with the
definition (e.g., the subformula order in the definition of
$\models$). In particular, application of a rule deriving some element
$x$ is delayed until no unsatisfied rule is left deriving a strictly
smaller object $y<x$.

It would be rather straightforward to formalize this idea for PC(ID)
if it was not that a PC(ID) definition $\D$ does not come with a
explicit order. Fortunately, there is a different way to make sure
that a rule can be safely applied, i.e., that later rule applications
during the inductive process will not falsify its condition. To do
this, we need to distinguish whether a defined atomic proposition has
been derived to be true, to be false or is still underived. E.g., once
$\valI \models \psi$ is derived to be true, we can safely apply the
rule for disjunctions and derive $\valI \models \psi\lor\phi$ to be
true, even $\valI \models\phi$ is still underived.  Likewise, we can
safely derive $\valI\not\models \psi\land\phi$ as soon as we found out
$\valI\not\models\psi$. Applying this criterion relies on the ability
to distinguish whether a defined atomic proposition (such as
``$\valI\models\psi$'') has been derived to be true, to be false or is
still underived, and whether a rule condition is certainly satisfied,
certainly dissatisfied or still unknown in such state.  This naturally
calls for a formalization of the induction process in a three-valued
setting where intermediate stages of the set in construction are
represented by three-valued sets instead of two-valued sets, and rules
are evaluated in these three-valued sets.

Below we present the formalization of the well-founded semantics
introduced in \cite{lpnmr/DeneckerV07}. Compared to the
original formalizations in \cite{VanGelder91,VanGelder93}, it is
geared directly at formalizing the inductive process as described
above, using concepts of three-valued logic. We start its presentation
by recalling some basic concepts of three-valued logic.

Consider the set of truth values $\{ \Tr, \Fa, \Un \}$. The {\em truth
  order} $\leq$ on this set is induced by $\Fa \leq \Un \leq \Tr$ and
the {\em precision order} $\leq_p$ is induced by $\Un \leq_p \Fa$ and
$\Un \leq_p \Tr$. Define $\Fa^{-1} = \Tr$, $\Un^{-1} = \Un$ and
$\Tr^{-1} = \Fa$.

Let $\voc$ be a propositional vocabulary. A three-valued
$\voc$-interpretation, also called a $\voc$-valuation, is a function
$I$ from $\voc$ to the set of truth values $\{ \Tr, \Fa, \Un \}$. An
interpretation is called two-valued if it maps no atoms to
$\Un$. Given two disjoint vocabularies $\voc$ and $\voc'$, a
$\voc$-interpretation $I$ and a $\voc'$-interpretation $I'$, the $\voc
\cup \voc'$-interpretation mapping each element $P$ of $\voc$ to
$I(P)$ and each $P \in \voc'$ to $I'(P)$ is denoted by $I+I'$. When
$\voc' \subseteq \voc$, we denote the restriction of a
$\voc$-interpretation $I$ to the symbols of $\voc'$ by
$\res{I}{\voc'}$. For a $\voc$-interpretation $I$, a truth value
$v$ 
and an atom $P \in \voc$, we denote by $I[P/v]$ the
$\voc$-interpretation that assigns $v$ to $P$ and corresponds to $I$
for all other atoms. We extend this notation to sets of atoms. Both
truth and precision order can be extended to an order on all
$\voc$-interpretations by $I \leq J$ if for each atom $P \in \voc$,
$I(P) \leq J(P)$ and $I \leq_p J$ if for each atom $P \in \voc$, $I(P)
\leq_p J(P)$.

A three-valued interpretation $I$ on $\voc$ can be extended to all PC-formulas over $\voc$ by induction on the subformula order:
\begin{itemize}
\item $\tf{P}{I} = I(P)$ if $P\in\voc$;
\item $\tf{(\varphi \land \psi)}{I} = min_\leq(\{\tf{\varphi}{I},\tf{\psi}{I}\})$;
\item $\tf{(\varphi \lor \psi)}{I} = max_\leq(\{\tf{\varphi}{I},\tf{\psi}{I}\})$;
\item $\tf{(\neg\varphi)}{I} = (\tf{\varphi}{I})^{-1}$.
\end{itemize}

The following proposition states a well-known monotonicity property
with respect to the precision order.

\begin{proposition} \label{prop:propertyofprecisionorder}
Let $\varphi$ be a PC-formula over $\voc$ and $I,J$ be three-valued $\voc$-interpretations such that $I \leq_p J$. Then $\varphi^{I} \leq_p \varphi^{J}$.
\end{proposition}

Another well-known proposition states a monotonicity property
with respect to the truth order.

\begin{proposition} \label{prop:propertyoftruthorder} Let $\varphi$ be
  a PC-formula over $\voc$ and $I, J$ be three-valued
  $\voc$-interpretations such that if $P^{I} < P^{J}$, then $P$ only
  occurs positively in $\varphi$ and if $P^{I} > P^{J}$ then $P$ only
  occurs negatively in $\varphi$. Then $\varphi^{I} \leq \varphi^{J}$.
\end{proposition}

The above properties about the precision and truth order will be applied frequently in the proofs in Section~\ref{sec:result}. For brevity,
we will not mention them explicitly in the remainder of the paper.

We now define the semantics of definitions. Let $\D$ be a definition
over $\voc$ and $I_O$ a two-valued $\openp{\D}$-interpretation, i.e.,
an interpretation of all open symbols of $\D$. Consider a
sequence of three-valued $\voc$-interpretations $(I^{n})_{n \geq 0}$
extending $I_O$ such that $I^{0}(P) = \Un$ for every $P \in
\defp{\D}$, and for every natural number $n$, $I^{n+1}$ relates to
$I^{n}$ in one of the following ways:
\begin{enumerate}
\item \label{itwell:1} $I^{n+1} = I^{n}[P/\Tr]$ where $P$ is a defined atom such that $P^{I^{n}}=\Un$ and for some rule $P \rul \varphi \in \D,
   \varphi^{I^{n}} = \Tr$.
\item \label{itwell:2} $I^{n+1} = I^{n}[U/\Fa]$, where $U$ is a non-empty set of defined atoms, such that for each $P \in U$, $I^{n}(P) = \Un$ and for each
   rule $P \rul \varphi \in \D$, $\tf{\varphi}{I^{n+1}} = \Fa$.
\end{enumerate}

The first derivation rule~\ref{itwell:1} derives true atoms and is a
straightforward formalization of the principle explained in the
beginning of this section. The second derivation rule~\ref{itwell:2} is less obvious
and serves to derive falsity of defined atoms. Let us first consider a
more obvious special case that is subsumed by rule~\ref{itwell:2}:

\begin{enumerate}\setcounter{enumi}{2}
\item \label{itwell:3} $I^{n+1} = I^{n}[P/\Fa]$ where $P$ is a defined
  atom such that $I^{n}(P) = \Un$ and for each rule $P\rul\varphi \in
  \D$, $\tf{\varphi}{I^{n}} = \Fa$.
\end{enumerate}
This rule expresses that if the body of each rule that could derive
$P$ is certainly false at stage $n$, then $P$ can be asserted to be
false at stage ${n+1}$.  This is a special case of the
rule~\ref{itwell:2}. Indeed, taking $U=\{P\}$, we have for each
$P\rul\varphi \in \D$ that $\Fa= \tf{\varphi}{I^{n}} \leq_p
\tf{\varphi}{I^{n}[U/\Fa]} = \tf{\varphi}{I^{n+1}} = \Fa$.

The stronger derivation rule~\ref{itwell:2} expresses that the atoms
in a set $U$ consisting of underived defined atoms can be turned to
false if the assumption that they are all false suffices to dissatisfy
the condition of each rule that could produce an element of $U$. A set
$U$ as used in this rule corresponds exactly to an {\em unfounded set}
as defined in ~\cite{VanGelder91}.  The rationale behind this
derivation rule and the link with informal induction is that when $U$
is an unfounded set at stage $n$ then none of its atoms can be derived
anymore at later stages of the construction process (using derivation
rule~\ref{itwell:1}).  To see this, assume towards contradiction that
at some later stage $>n$, one or more elements of $U$ could be derived
to be true, and let $P$ be the first atom that could be derived, say
at stage $m>n$. At stage $m$, it holds for each $Q\in U$ that
$I^m(Q)=\Un$ and for some rule $P\rul\varphi\in\D$,
$\varphi^{I^m}=\Tr$. But $I^n[U/\Fa] \leq_p I^m[U/\Fa] \geq_p I^m$ and
hence, $\Fa = \tf{\varphi}{I^n[U/\Fa]}\leq_p \tf{\varphi}{I^m[U/\Fa]}
\geq_p \tf{\varphi}{I^m}= \Tr$ and this yields a contradiction. Thus,
the derivation rule~\ref{itwell:2} correctly concludes that the atoms
in $U$ are no longer derivable through rule application. This
derivation rule is needed to derive, e.g., falsity of all atoms not in
the least fixpoint of a monotone  definition, which is something that
cannot be derived in general by the rule~\ref{itwell:3}.

We call a sequence as defined above a {\em well-founded induction}. A
well-founded induction is {\em terminal} if it cannot be extended
anymore. It can be shown that each terminal well-founded induction is
a sequence of increasing precision and its limit is the {\em
  well-founded partial interpretation} of $D$ extending
$I_O$~\cite{lpnmr/DeneckerV07}. We denote the well-founded partial
interpretation of $D$ extending $I_O$ by $I_{O}^{\D}$.

We define that $D^I = \Tr$ if $I = {(\res{I}{\openp{\D}})}^{\D}$ and $I$
is two-valued. Otherwise, we define $\D^I = \Fa$. Adding this as a new
base case to the definition of the truth function of formulas, we can
extend the truth function inductively to all PC(ID)-formulas.

We are now ready to define the semantics of PC(ID). For an arbitrary
PC(ID)-formula $\varphi$, we say that an interpretation $I$ satisfies
$\varphi$, or $I$ is a model of $\varphi$, if $I$ is two-valued and
$\varphi^I = \Tr$. As usual, this is denoted by $I \models
\varphi$. $I$ satisfies (is a model of) a PC(ID) theory $T$ if $I$
satisfies every $\varphi \in T$.

A definition lays a functional relation between the interpretation of
the defined symbols and those of the open symbols. In particular, two
models of a definition differ on the open symbols. A model of a
monotone definition is the $\leq$-least interpretation satisfying the
rules of the definition (interpreted as material implications) given a
fixed interpretation of the open symbols, as desired.  Also, the
semantics of PC(ID) is two-valued and extends the standard semantics
of propositional logic. A three-valued interpretation $I$ is never a
model of a definition, not even if it is a well-founded partial
interpretation of the definition.

\begin{example} \label{ex:definition}
Consider the following definition:
\[D = \defin{P \rul Q\\
             Q \rul P}. \]
Then $\openp{D} = \emptyset$ and $\defp{D} = \{ P, Q \}$. There are no open symbols and there is only one model of $\D$, namely the interpretation mapping both $P$ and $Q$ to $\Fa$.
\end{example}

\subsection{Where  the informal semantics breaks} \label{sec:InformalSemBreaks}

The informal semantics of a PC(ID) rule set as an inductive definition
breaks in some cases. Examples are non-monotone rule sets with
recursion over negation such as
\[ \defin{P \rul \neg P}\]
or
\[ \defin{P \rul \neg Q\\Q \rul \neg P}\]
Their (unique) well-founded partial interpretation is not two-valued,
and hence, these definitions have no model and are inconsistent in PC(ID).

The restriction to two-valued well-founded partial models was imposed
to enforce the view that a well-designed definition $\D$ ought to
define the truth of all its defined atoms, i.e., the inductive process
should be able to derive truth or falsity of all defined atoms. This
motivates the following concept.

\begin{definition}[Totality,\cite{DeneckerT08}] \label{def:totality}
  Let $I_O$ be a two-valued interpretation of $\openp{\D}$. A
  definition $\D$ is {\em total} in $I$ if $I_O^\D$ is two-valued. The
  definition $D$ is {\em total in the context of a theory} $T$ if $\D$
  is total in $\res{M}{\openp{\D}}$, for each model $M$ of $T$.  A
  definition $\D$ is \emph{total} if it is total in every two-valued
  interpretation $I_O$ of its open atoms.
\end{definition}

A simple and very general syntactic criterion that guarantees that a
definition is total can be phrased in terms of the dependency relation
$\prec$ of $\D$.  A definition $\D$ is {\em stratified} if for each
rule $P\rul\varphi$, for each symbol $Q$ with a negative occurrence in
$\varphi$, $P\not\prec Q$. This means that the definition of $Q$ does
not depend on $P$.

\begin{proposition}[\cite{VanGelder91}] If $\D$ is stratified then
  $\D$ is total.
\end{proposition}
Observe that a stratified definition formally satisfies the (informal)
condition that was stated for iterated inductive definitions early in
this section. The well-founded semi-order underlying an iterated
inductive definition is nothing else than the reflexive closure
$\preceq$ of $\prec$. Atoms $Q$ with a positive occurrence in the body
of a rule deriving $P$ satisfy $Q\preceq P$; those with a negative
occurrence satisfy $Q\preceq P$ and $P\not\preceq Q$. Hence, such
rules effectively derive $P$ given the presence of less or equivalent
atoms and the absence of strictly less atoms in the defined
valuation. The well-founded model of such definitions is two-valued
and corresponds exactly to the structure obtained by the construction
described in Section~\ref{sec:sem} for (informal) inductive
definitions. Thus, the well-founded semantics correctly
formalizes the informal semantics of inductive definitions, and
correctly constructs the (informally) defined relations without
knowing the underlying (semi-)order of the definition.

Although the class of stratified definitions is large and comprises
almost all ``practical'' PC(ID) definitions that we  encountered
in applications, there are intuitively sensible definitions which are
total but not stratified.

\begin{example} A software system consists of two servers $S1$ and
  $S2$ that provide identical services.  One server acts as master and
  the other as slave, and these roles are assigned on the basis of
  clear (but irrelevant) criterion that can be expressed in the form
  of a set of defining rules for the predicate $Master(s)$.  Clients
  can request services $x$. The master makes a selection among these
  requests on the basis of a clear (but irrelevant) criterion
  expressed in a definition of $Criterion(x)$.  The slave fulfills all
  requests that are not accepted by the master. Here is the core of a
  (predicate) definition:
  \[ \defin{ Criterion(x) \rul \dots \\
    Master(s) \rul \dots \\
    Slave(s) \rul \neg Master(s)\\
    Accepts(x, m) \rul Request(x) \land Master(m) \land Criterion(x)\\
    Accepts(x, s) \rul Request(x) \land Slave(s) \land \exists m
    (Master(m) \land \neg Accepts(x,m)) }\] The (propositionalisation
  of the) definition is not stratified since the last rule creates a
  negative dependency between $Accepts(x,S1)$ and
  $Accepts(x,S2)$. Yet, since no server can be both master and slave,
  this recursion is broken ``locally'' in each model. This is a total,
  albeit unstratified definition of the predicate $Accepts$ that
  correctly implements the informal specification.
\end{example}

The proof system for PC(ID), as presented below, is sound and complete
with respect to all PC(ID) theories containing only total definitions,
and hence to any fragment of PC(ID) that enforces totality of the
allowed definitions.

\ignore{
\begin{example} \label{ex:nontotaldefinition}
Consider the definition as follows:
\[ D = \defin{P \rul O\\
  Q \rul P \land \neg Q}. \] This definition is total in the
interpretation mapping $O$ to $\Fa$ but not total in the
interpretation mapping $O$ to $\Tr$. Thus, it is not a total
definition.
\end{example}
}

\section{\ps: A Proof system for PC(ID)}\label{sec:dedsys}
In this section we formulate a proof system, \ps, for the logic PC(ID) in the sequent calculus style originally developed by Gentzen in 1935~\cite{Gentzen35}. Our system can be seen essentially as a propositional part of classical sequent calculus adaptation of inference rules for definitions. We give the proof rules of \ps, which are the rules of Gentzen's original sequent calculus for propositional logic, augmented with rules for introducing defined atoms on the left and right of sequents, a rule for inferring the non-totality of definitions and a rule for introducing definitions on the right of sequents.


First, we introduce some basic definitions and notations. Let capital Greek letters $\Gamma, \Delta, \ldots$ denote finite (possibly empty)
sets of PC(ID)-formulas.
$\Gamma, \Delta$ denotes $\Gamma \cup \Delta$. $\Gamma, \varphi$ denotes $\Gamma \cup \{ \varphi \}$. By $\bigwedge \Gamma$, respectively $\bigvee \Gamma$, we denote the conjunction, respectively disjunction of all formulas in $\Gamma$. By $\neg \Gamma$, we denote the set obtained by taking the negation of each formula in $\Gamma$. By $\Gamma \setminus \Delta$, we denote the set obtained by deleting from $\Gamma$ all occurrences of formulas that occur in $\Delta$. $\Gamma$ is said to be {\em consistent} if there is no formula $\varphi$ such that both $\varphi$ and $\neg \varphi$ can be derived from $\Gamma$.

A \emph{sequent} 
is an expression of the form $\Gamma \la \Delta$.
$\Gamma$ and $\Delta$
are respectively called the {\em antecedent} and {\em succedent} of the sequent and each formula in $\Gamma$ and $\Delta$ is called a {\em sequent
formula}.
In general, a formula $\varphi$ occurring as part of a sequent denotes the set $\{ \varphi \}$. We will denote sequents by $S, S_1, \ldots$.
A sequent $\Gamma \la \Delta$ is \emph{valid}, denoted by $\models \Gamma \la \Delta$, if every model of $\bigwedge \Gamma$ satisfies $\bigvee \Delta$.
A \emph{counter-model} for $\Gamma \la \Delta$ is an interpretation $I$ such that $I \models \bigwedge \Gamma$ but $I \not\models \bigvee \Delta$. The sequent $\Gamma \la $ is equivalent to $\Gamma \la \bot$ and $\la \Delta$ is equivalent to $\top \la \Delta$, where $\bot, \top$ are logical constants denoting {\em false} and {\em true}, respectively.

An {\em inference rule} is an expression of the form \[ \inferencerule{S_1; \ldots; S_n}{S} \quad n \geq 0 \] where $S_1, \ldots, S_n$ and $S$ are sequents. Each $S_i$ is called a \emph{premise} of the inference rule, $S$ is called the \emph{consequence}. Intuitively, an inference rule means that $S$ can be inferred, given that all $S_1, \ldots, S_n$ are already inferred.

The {\em initial sequents}, or {\em axioms} of \ps are all sequents of the form
\[\Gamma, A \la A, \Delta  \mbox { \ \ or \ \ } \bot \rightarrow \Delta  \mbox { \ \ or \ \ } \Gamma \rightarrow \top \]
where $A$ is any PC(ID)-formula, $\Gamma$ and $\Delta$ are arbitrary 
sets of PC(ID)-formulas.

The inference rules for \ps consist of {\em structural} rules, {\em logical} rules and {\em definition} rules. The structural and logical rules, which follow directly the propositional inference rules in Gentzen's original sequent calculus for first-order logic \gs, deal with the propositional part of PC(ID) and are given as follows, in which $A, B$ are any PC(ID)-formulas and $\Gamma, \Delta$ are arbitrary sets of PC(ID)-formulas.

\subsubsection{Structural rules}
\begin{itemize}
   \item Weakening rules
  \[ \text{left:} \ \inferencerule{\Gamma\rightarrow \Delta}{A, \Gamma \rightarrow \Delta}; \ \ \text{right:} \ \inferencerule{\Gamma \rightarrow
  \Delta}{\Gamma \rightarrow \Delta, A}.\]
   \item Contraction rules
  \[ \text{left:} \ \inferencerule{A, A, \Gamma \rightarrow \Delta}{A, \Gamma \rightarrow \Delta}; \ \ \text{right:} \ \inferencerule{\Gamma
  \rightarrow \Delta, A, A}{\Gamma \rightarrow \Delta, A}.\]
   \item Cut rule
  \[ \inferencerule{\Gamma \rightarrow \Delta, A; \ \ A, \Gamma \rightarrow \Delta}{\Gamma \rightarrow \Delta}.\]
\end{itemize}

\subsubsection{Logical rules}
\begin{itemize}
   \item $\neg$ rules
      \[ \text{left:} \inferencerule{\Gamma\rightarrow \Delta, A}{\neg A, \Gamma \rightarrow \Delta}; \ \ \text{right:} \inferencerule {A, \Gamma
      \rightarrow \Delta}{\Gamma \rightarrow \Delta, \neg A}.\]
   \item $\land$ rules
      \[  \text{left:} \inferencerule{A, B, \Gamma \rightarrow \Delta}{A \land B, \Gamma \rightarrow \Delta}; \ \ \text{right:} \inferencerule{ \Gamma
      \rightarrow \Delta, A ; \ \Gamma \rightarrow \Delta, B}{\Gamma \rightarrow \Delta, A \land B}.\]
   \item $\lor$ rules
      \[ \text{left:} \inferencerule{A, \Gamma\rightarrow \Delta; \ B, \Gamma \rightarrow \Delta }{A \lor B, \Gamma \rightarrow \Delta}; \ \
      \text{right:} \inferencerule{\Gamma \rightarrow \Delta, A, B}{\Gamma \rightarrow \Delta, A \lor B}.\]
\end{itemize}

\quad \quad \quad

Our deductive system \ps is then obtained from the propositional part of \gs by adding inference rules for definitions. The definition rules of \ps consist of the {\em right definition} rule, the {\em left definition} rule, the {\em non-total} definition rule and the {\em definition introduction} rule. Without loss of generality, in what follows we assume that there is only one rule with head $P$ in a definition $\D$ for every $P \in \defp{\D}$. We refer to this rule as {\em the rule for $P$ in $D$} and denote it by $P \rul \varphi_{P}$.

\subsubsection{Right definition rule for $P$.}
The \emph{right definition rule} introduces defined atoms in the succedents of sequents. It allows inferring the truth of a defined atom from a definition $\D$ and is therefore closely related to the derivation rule~\ref{itwell:1} for extending a well-founded induction. Let $\D$ be a definition and $P$ a defined atom of $D$. The right definition rule for $P$ is given as follows.
\[\inferencerule{\Gamma \rightarrow  \Delta, \varphi_{P} }{D, \Gamma \rightarrow \Delta, P} \]
where $\Gamma$ and $\Delta$ are arbitrary 
sets of PC(ID)-formulas.

We illustrate this inference rule with an example.
\begin{example}\label{ex:rightdefinitionrule}
Consider the definition
\[
   D = \defin{ P \leftarrow P \land \neg Q\\
               Q \leftarrow \neg P}.
\]
The instance of the right definition rule for $P$ is
\[ \inferencerule{\Gamma \la \Delta, P \land \neg Q}{D, \Gamma \rightarrow \Delta, P},\]
and the instance of the right definition rule for $Q$ is
\[\inferencerule{\Gamma \rightarrow \Delta, \neg P}{D, \Gamma \rightarrow \Delta, Q}. \]
\end{example}

\subsubsection{Left definition rule for $P_i \in U$.}
The \emph{left definition rule} introduces defined atoms in the antecedents of sequents. It allows inferring the falsity of a defined atom from a definition $\D$ and is therefore closely related to the second derivation rule~\ref{itwell:2} for extending a well-founded induction.

We first introduce some notations. Given a set $U$ of atoms, let $U^{\triangleright}$ be a set consisting of one new atom $P^{\triangleright}$ for every $P\in U$. The vocabulary $\voc$ augmented with these symbols is denoted by $\voc^{\triangleright}$. Given a PC-formula $\varphi$, $\varphi^{\triangleright}$ denotes the formula obtained by replacing all positive occurrences of an atom $P \in U$ in $\varphi$ by $P^{\triangleright}$.  We call $\varphi^{\triangleright}$ the {\em renaming} of $\varphi$ with respect to $U$. For a set of PC-formulas $F$, $F^{\triangleright}$ denotes $\{ \varphi^{\triangleright} \mid \varphi \in F \}$. For arbitrary PC-formula $\varphi$, by $\neg \varphi^{\triangleright}$, we mean $\neg (\varphi^{\triangleright})$.

Let $\D$ be a definition over $\voc$ and $U$ a non-empty set of atoms such that $U \subseteq \defp{D}$. Denote by $\neg U^{\triangleright}$ the set $\{ \neg P^{\triangleright}  | P \in U \}$. Let $\Gamma$ and $\Delta$ be
sets of PC(ID)-formulas over $\voc$. The left definition rule for every $P_i \in U$ is given as follows, where $U = \{P_1, \ldots, P_n \}$.

$$\inferencerule{\neg U^{\triangleright}, \Gamma \la \Delta, \neg \varphi^{\triangleright}_{P_{1}}; \ldots; \neg U^{\triangleright}, \Gamma \la \Delta, \neg \varphi^{\triangleright}_{P_n}}{P_i, D, \Gamma \la \Delta}.$$

\quad \quad

Actually, in the left definition rule, the set of atoms $U$ is a candidate unfounded set of $D$.

We illustrate this inference rule with an example.
\begin{example} \label{ex:leftdefinition}
Given a definition $D = \defin{
P \rul P \land \neg Q\\
Q \leftarrow Q
},$
\begin{itemize}
   \item $U = \{ P \}$, the instance of the left definition rule for $P \in U$ is
\[\inferencerule{\neg P^{\triangleright}, \Gamma \la \Delta, \neg (P^{\triangleright} \land \neg Q)}{P, D, \Gamma \la \Delta}\]
   \item $U = \{ Q \}$, the instance of the left definition rule for $Q \in U$ is
\[ \inferencerule{\neg Q^{\triangleright}, \Gamma \la \Delta, \neg Q^{\triangleright}}{Q, D, \Gamma \la \Delta} \]
   \item $U = \{ P, Q \}$, the instance of the left definition rule for $P \in U$ is
\[ \inferencerule{\neg P^{\triangleright}, \neg Q^{\triangleright}, \Gamma \la \Delta, \neg (P^{\triangleright} \land \neg Q); \ \ \neg P^{\triangleright}, \neg Q^{\triangleright}, \Gamma \la \Delta, \neg Q^{\triangleright}}{P, D, \Gamma \la \Delta} \]
   \item $U = \{ P, Q \}$, the instance of the left definition rule for $Q \in U$ is
\[ \inferencerule{\neg P^{\triangleright}, \neg Q^{\triangleright}, \Gamma \la \Delta, \neg (P^{\triangleright} \land \neg Q); \ \ \neg P^{\triangleright}, \neg Q^{\triangleright}, \Gamma \la \Delta, \neg Q^{\triangleright}}{Q, D, \Gamma \la \Delta}. \]
\end{itemize}
\end{example}

\subsubsection{Non-total definition rule for $D$.}
The {\em non-total definition rule} allows inferring the non-totality
of a definition $D$. We introduce some notations. Let $D$ be a
definition over $\voc$ and $V$ a non-empty set of atoms such that $V \subseteq
\defp{D}$. Denote by $\voc^{\diamond}$ the vocabulary $\voc \cup V^{\triangleright} \cup V^{\diamond}$,
where both $V^{\triangleright}$ and $V^{\diamond}$ are sets of new and different renamings $P^{\triangleright}
$ and $P^{\diamond}$ of all symbols $P$ of $V$. Denote by $\varphi^{\diamond}$ the
formula obtained by replacing each positive occurrence of each $P \in
V$ in $\varphi$ by $P^{\triangleright}$ and each negative occurrence of each $P \in V$ in $\varphi$ by $P^{\diamond}$.
Denote by $D^{\diamond}$ the definition $\{ P^{\triangleright} \rul \varphi_{P}^{\diamond} \mid P \in V$ and
$P \rul \varphi_{P} \in D \}$ over the new vocabulary $\voc^{\diamond}$. Let
$\Gamma$ and $\Delta$ be 
sets of PC(ID)-formulas over $\voc$.
Then the non-total definition rule for $D$ is given as follows.
$$\inferencerule{V^{\diamond}, D^{\diamond}, \Gamma \la \Delta, \bigwedge \neg V^{\triangleright}; \ \neg V^{\diamond}, D^{\diamond}, \Gamma \la \Delta, \bigwedge V^{\triangleright}}{D, \Gamma \la \Delta}$$

We illustrate this inference rule with an example.
\begin{example} \label{ex:nontotaldefinitionrule}
Given a definition $D = \defin{P \rul P \land \neg Q\\ Q \rul \neg Q \land R\\ R \rul \neg R}$, $V = \{ Q, R \}$ and $\Gamma$ and $\Delta$ empty sets. Then the instance of the non-total definition for $D$ is
\[\inferencerule{Q^{\diamond}, R^{\diamond}, D^{\diamond} \la \neg Q^{\triangleright} \land \neg R^{\triangleright}; \ \neg Q^{\diamond}, \neg R^{\diamond}, D^{\diamond} \la Q^{\triangleright} \land R^{\triangleright}}{D \la }, \]
where $D^{\diamond} = \defin{Q^{\triangleright} \rul \neg Q^{\diamond} \land R^{\triangleright}\\ R^{\triangleright} \rul \neg R^{\diamond}}$.
\end{example}

\quad \quad \quad \quad

For the intuition behind the non-total definition rule, we point the readers to \cite{DeneckerT08} and Section~\ref{sec:InformalSemBreaks} where the cause of the non-totality of a definition is explained.

We do not have an inference rule to prove totality of all definitions in the context of a certain set $\Gamma$ of PC(ID)-formulas. Such an inference rule would involve proving that each model of $\Gamma$ can be extended to a model of the definition. In fact, we cannot even formulate this condition as a sequent.

\subsubsection{Definition introduction rule for $D$.}
The three definitional inference rules introduced so far, introduce  a definition in the antecedent of the consequence. Hence, none of these rules can be used to infer that under certain conditions a definition holds. The {\em definition introduction rule} allows inferring the truth of a total definition from PC(ID)-formulas.

We introduce some notations. Let $D$ be a total definition. Denote by $P'$ a new defined atom for each $P \in \defp{D}$. Denote by $\voc'$ the vocabulary $\voc \cup \{ P' \mid P \in \defp{D} \}$. Denote by $D'$ the definition over the new vocabulary $\voc'$ obtained by replacing each occurrence of each defined symbol $P$ in $D$ by $P'$.
Let $\Gamma$ and $\Delta$ be sets of PC(ID)-formulas over the old vocabulary $\voc$.
The definition introduction rule for $D$ is given as follows, where $P_1, \ldots, P_n$ are all defined atoms of $D$.
$$\inferencerule{D', \Gamma \la \Delta, P'_1 \equiv P_1; \ldots ; D', \Gamma \la \Delta, P'_n \equiv P_n}{\Gamma \la \Delta, D}$$

We illustrate this inference rule with an example.
\begin{example} \label{ex:defintroduction}
Given a definition $D = \defin{P \rul O \\
                               Q \rul Q \land P}$, $\Gamma = \{ O, P, \neg Q \}$ and $\Delta$ an empty set. Then the
instance of the definition introduction rule for $D$ is
\[ \inferencerule{D', O, P, \neg Q \la P' \equiv P; \ \ D', O, P, \neg Q \la Q' \equiv Q}{O, P, \neg Q \la \D} ,\]
where $D' = \defin{P' \rul O\\
                   Q' \rul Q' \land P'}$.
\end{example}

The inference rule proposed here has a definition in the succedent of its premise and hence,
allows to infer the truth of a definition.
Unfortunately, this rule is only sound given that the inferred definition is total.
We will give an example to show that the definition introduction rule is not sound given that the inferred definition is non-total right after proving the soundness of this inference rule.

\subsubsection{Proofs of  PC(ID).}
We now come to the notion of an \emph{\ps-proof} for a sequent.
\begin{definition} \label{def:provable}
An \ps-{\em proof} for a sequent $S$, is a tree $T$ of sequents with root $S$. Moreover, each leaf of $T$ must be an axiom and for each interior node $S'$ there exists an instance of an inference rule such that $S'$ is the consequence of that instance while the children of $S'$ are precisely the premises of that instance. $T$ is often called a {\em proof tree} for $S$. A sequent $S$ is called {\em provable} in \ps, or \ps-{\em provable}, if there is an \ps-proof for it.
\end{definition}
\begin{example} \label{ex:aproof}
Given a definition
$D = \defin{P \rul O\\
         Q \rul Q \land P}$,
the following is an \ps-proof for $O, D \la P \land \neg Q$.
\begin{prooftree}
\AxiomC{$O \la O$}
\LeftLabel{right definition rule}
\UnaryInfC{$O, D \la P$}
				\AxiomC{$Q^{\triangleright} \la Q^{\triangleright}$}
				\RightLabel{left $\neg$}
				\UnaryInfC{$\neg Q^{\triangleright}, Q^{\triangleright} \la$}
				\RightLabel{left weakening}
				\UnaryInfC{$\neg Q^{\triangleright}, Q^{\triangleright}, P \la$}
				\RightLabel{left $\land$}
				\UnaryInfC{$\neg Q^{\triangleright}, Q^{\triangleright} \land P \la$}
				\RightLabel{right $\neg$}
				\UnaryInfC{$\neg Q^{\triangleright} \la \neg (Q^{\triangleright} \land P)$}
				\RightLabel{left definition rule}
				\UnaryInfC{$Q, D \la$}
				\RightLabel{right $\neg$}
				\UnaryInfC{$D \la \neg Q$}
				\RightLabel{left weakening}
				\UnaryInfC{$O, D \la \neg Q$}
				\RightLabel{right $\land$}
\BinaryInfC{$O, D \la P \land \neg Q$}
\end{prooftree}
\end{example}

\section{Main results}\label{sec:result}

In this section, we will prove that the deductive system \ps is sound and complete for a slightly restricted fragment of PC(ID), which can be viewed as main theoretical results of this paper.

\subsection{Soundness}

To prove the soundness of \ps, it is sufficient to prove that all axioms of \ps are valid and that every inference rule of \ps is sound, i.e. if all premises of an inference rule are valid then the consequence of that rule is valid. It is trivial to verify that the axioms are valid and that the structural and logical rules are sound (see e.g. \cite{Szabo69,Takeuti75}). Hence, only the soundness of the right definition rule, the left definition rule, the non-total definition rule and the definition introduction rule must be proved.

\begin{lemma} \label{lem:A}
Let $I$ be a model of $D$ and $P$ a defined atom of $D$. Then $I\models P$ if and only if $I \models \varphi_{P}$.
\end{lemma}

\begin{proof}
Because $I$ is a model of $D$, there exists a terminal well-founded induction $(I^n)_{n \leq \xi}$ for $D$ with the limit $I^\xi = I$.

(if part) Assume that $I \models \varphi_{P}$. The sequence $(I^{n})_{n \leq \xi}$ is strictly increasing in precision, hence there is no $n \leq \xi$ such that $\varphi_{P}^{I^n} = \Fa$. As such, for every $n \leq \xi$, $P^{I^n} \not = \Fa$. Therefore, $P^I \not = \Fa$ and because $I$ is two-valued, we can conclude $P^I = \Tr$.

(only if part) Assume that $I \models P$. Thus, for some $n < \xi$, $P^{I^n} = \Un$ and $P^{I^{n+1}} = \Tr$. Hence, $\varphi_{P}^{I^n} = \Tr$. Because the sequence $(I^n)_{n \leq \xi}$ is strictly increasing in precision, we have $\varphi_{P}^{I} = \Tr$.
\end{proof}

\begin{lemma}[Soundness of the right definition rule]\label{lem:disound}
Let $D$ be a definition and $P$ a defined atom of $D$. If $\models \Gamma \la \Delta, \varphi_P$, then $\models D, \Gamma \la \Delta, P$.
\end{lemma}

\begin{proof}
  Assume $\models \Gamma \la \Delta, \varphi_{P}$ but $\not\models
 D, \Gamma \la \Delta, P$. Then there exists a counter-model $I$ for
  $D, \Gamma \la \Delta, P$ which satisfies $D, \bigwedge \Gamma, \neg \bigvee \Delta$ and $ \neg
  P$.  It follows from the first assumption
  that $I \models \varphi_{P}$, and hence, by  Lemma \ref{lem:A}, $I \models P$, a contradiction.
\end{proof}

\begin{lemma}[Soundness of the left definition rule]\label{lem:desound}
   Let $D$ be a definition
  and $U$ be a non-empty subset of $\defp{D}$.  If for every $P\in U$, it holds that $\models \neg
  U^{\triangleright}, \Gamma \la \Delta, \neg \varphi_{P}^{\triangleright}$, then for all $P \in U$, it holds that
  $\models P, D, \Gamma \la \Delta$.
\end{lemma}

\begin{proof}
  Assume $\models \neg U^{\triangleright}, \Gamma \la \Delta, \neg\varphi_{P}^{\triangleright}$ for every
  $P \in U$, but $\not \models P, D, \Gamma
  \la \Delta$ for some $P \in U$. Then there exists a model $I$ of
  $D$, $\bigwedge \Gamma$ and $\neg \bigvee\Delta$ satisfying at least
  one $P\in U$. Furthermore, by Lemma \ref{lem:A},
  it holds that $I \models \varphi_{P}$.
  We select this $P$
  in the following way. Let
  $(I^{n})_{n \leq \xi}$ be a terminal well-founded
  induction for $D$ with limit $I^{\xi} = I$. Let $n$ be the
  smallest $n \leq \xi$ such that for some $Q \in U$,
  $Q^{I^{n}} = \Un$ and $Q^{I^{n + 1}} = \Tr$. By selection
  of $n$, there is a unique $P\in U$ such
  that $P^{I^n}=\Un$, $I^n \models \varphi_{P}$ and $P^{I^{n
      + 1}} = \Tr$. Consider this $P$ and $\varphi_{P}$.

  On the one hand, it holds that $I\models \varphi_{P}$. On the other
  hand, consider the interpretation $I^{\triangleright} = I[U^{\triangleright}/ \Fa]$. It is clear that
  $I^{\triangleright}$ satisfies $\neg U^{\triangleright}$, $\bigwedge \Gamma$ and $\neg
  \bigvee\Delta$.  Hence, by the first assumption, it holds that $I^{\triangleright}
  \models \neg \varphi_{P}^{\triangleright}$.  We will derive a contradiction from this.

  Observe that by our choice of $n$, for each $Q \in U$,
  $Q^{I^{n}} = \Fa$ or $Q^{I^{n}} = \Un$. Denote by
  ${I^{n}}^{\triangleright}$ the interpretation that assigns $Q^{I^{n}}$ to $Q^{\triangleright}$
  for every $Q\in U$ and corresponds to $I^{n}$ on all other
  atoms. There are two simple observations that can be made about
  ${I^{n}}^{\triangleright}$:
\begin{itemize}
\item ${I^{n}}^{\triangleright} \leq_p I^{\triangleright}$: indeed, $I^n \leq_p I$ and for each $Q^{\triangleright}\in U^{\triangleright}$, ${Q^{\triangleright}}^{I^{\triangleright}} = \Fa \geq_p {Q^{\triangleright}}^{{I^{n}}^{\triangleright}} = Q^{I^{n}} = \Fa \mbox{ or } \Un$.
\item $(\varphi_{P}^{\triangleright})^{{I^{n}}^{\triangleright}} = \varphi_{P}^{I^{n}} = \Tr$: obvious from the construction of ${I^{n}}^{\triangleright}$ and $\varphi_{P}^{\triangleright}$.
\end{itemize}
Combining these results, we obtain $\Tr = (\varphi_{P}^{\triangleright})^{{I^{n}}^{\triangleright}} \leq_p (\varphi_{P}^{\triangleright})^{I^{\triangleright}} = \Fa$. This is the desired contradiction.
\end{proof}

Having the soundness of the left definition rule, we can explain the
introduction of renaming formulas in the left definition rule.  Consider the left definition rule of the following
form:
\begin{equation}\label{eq:wrong}
\inferencerule{\neg P_1, \ldots, \neg P_n, \Gamma \la \Delta, \neg \varphi_{P_1}; \ldots; \neg P_1, \ldots, \neg P_n,
\Gamma \la \Delta, \neg \varphi_{P_n}}{P_i, D, \Gamma \la \Delta} \end{equation}
where $\{ P_1, \ldots, P_n \} \subseteq \defp{D}$ and $P_i$ is an arbitrary defined atom in $\{ P_1, \ldots, P_n \}$.

Intuitively, the above form of the left definition rule is exactly related to the
second derivation rule~\ref{itwell:2} of the well-founded induction and it is easier to be understood.
However, such an inference rule is not sound. For an arbitrary definition $D$ and any defined atom $P$ of $D$,
$D \la \neg P$ can be inferred applying this rule. We illustrate this with the next example.
\begin{example} \label{ex:necessarityofrenaming} Consider the following definition:
\[D = \defin{P \rul \top}.\]
Let $\Gamma = \{P\}$ and $\Delta$ be an empty set. Since $\neg P, P \la \neg \top$,
we can prove $D \la \neg P$ by using the inference rule \eqref{eq:wrong},
the right $\neg$ rule and the right contraction rule. However, for the same definition $D$ and
empty sets $\Gamma$ and $\Delta$, it is obvious that $D \la P$ can be inferred by using the right definition rule,
which derives a contradiction. Hence, the inference rule \eqref{eq:wrong} is not sound.
\end{example}

From the viewpoint of semantics, since the left definition rule corresponds to the second derivation rule~\ref{itwell:2} of the well-founded
induction, we have to adopt the approach of renaming to represent that the defined atoms of $U$ are unknown in $I^n$
and false in $I^{n + 1}$.

\begin{lemma} \label{lem:apropertyofwellfoundedmodel} Let $D$ be a
  definition, $I$ a model of $D$ and $U$ a non-empty subset of $\defp{D}$.
  If for every $P \in U$, it holds that $\varphi_{P}^{I[U/ \Fa]} = \Fa$, then $P^{I} = \Fa$ for all $P\in U$.
\end{lemma}

\begin{proof}
Assume that there exists a non-empty set $T$ satisfying that \begin{inlinenum}
  \item\label{(i)} $T \subseteq U$,
  \item\label{(ii)}$P^I = \Tr$ for each $P \in T$, and
  \item\label{(iii)} $P^I = \Fa$ for each $P \in U \setminus T$.
  \end{inlinenum}
Let $(I^n)_{n \leq \xi}$ be a terminal well-founded induction for $D$ with the limit $I^\xi = I$. Let $n$ be the smallest $n \leq \xi$ such that for some $Q \in T$, $Q^{I^n} = \Un$ and $Q^{I^{n+1}} = \Tr$. By selection of $n$, there is a unique $P \in T$ such that $P^{I^n} = \Un$, $\varphi_{P}^{I^n} = \Tr$ and $P^{I^{n+1}} = \Tr$. Consider this $P$ and $\varphi_{P}$.

Observe that by our choice of $n$, for each $Q \in T$, $Q^{I^n} = \Un$. Hence, for each $Q \in T$, $Q^{I^n} \leq_p Q^{I[U / \Fa]} = \Fa$. Because $I^n \leq_p I$, for each $Q \in \defp{D} \setminus T$, we have that $Q^{I[U / \Fa]} = Q^I \geq_p Q^{I^n}$. Combining these results, it is concluded that $I^{n} \leq_p I[U / \Fa]$. Therefore, we obtain that $\Tr = \varphi_{P}^{I^n} \leq_p \varphi_{P}^{I[U / \Fa]} = \Fa$, a contradiction. Hence, there is no $P \in U$ such that $P^I = \Tr$, which follows directly that $P^I = \Fa$ for all $P \in U$.
\end{proof}


\begin{lemma} [Soundness of the non-total definition rule]
\label{lem:nontotalsoundness}
If $\models V^{\diamond}, D^{\diamond}, \Gamma \la \Delta, \bigwedge \neg V^{\triangleright}$ and $\models
\neg V^{\diamond}, D^{\diamond}, \Gamma \la \Delta, \bigwedge V^{\triangleright}$, then $\models D, \Gamma \la
\Delta$.
\end{lemma}

\begin{proof}
Assume towards contradiction that
\begin{equation} \label{AssA}
\models V^{\diamond}, D^{\diamond}, \Gamma \la \Delta, \bigwedge \neg V^{\triangleright} \mbox{\ and \ }
\models \neg V^{\diamond}, D^{\diamond}, \Gamma  \la \Delta, \bigwedge V^{\triangleright} \mbox{\ but \ }
\not \models D, \Gamma \la \Delta.
\end{equation}
Then there exists a $\voc$-interpretation $I$ satisfying $D$, $\bigwedge \Gamma$ and $\neg \bigvee \Delta$. Consider the vocabulary $\tau^{\diamond} =
\tau\cup V^{\triangleright}\cup V^{\diamond}$.
$I$ can be expanded into two
$\tau^{\diamond}$-interpretations $I_{V^{\diamond}}$ and $I_{\neg V^{\diamond}}$ as follows:
$$I_{V^{\diamond}} = (I[V^{\diamond}/ \Tr])^{D^{\diamond}} \mbox{\ \ and \ \ } I_{\neg V^{\diamond}} =
(I[V^{\diamond}/ \Fa])^{D^{\diamond}}.$$
Since $D^{\diamond}$ is a positive definition, hence total definition with open symbols
$\tau\cup V^{\diamond}$, both interpretations are well-defined. Moreover they
obviously satisfy:
$$I_{V^{\diamond}} \models \bigwedge V^{\diamond}\land D^{\diamond} \land \bigwedge \Gamma \land \neg \bigvee \Delta
\mbox{\ \ \ and \ \ \ }I_{\neg V^{\diamond}} \models
\bigwedge \neg V^{\diamond}\land D^{\diamond} \land \bigwedge \Gamma \land \neg \bigvee \Delta.$$

By (\ref{AssA}), it follows that
\begin{equation} \label{PropC}
I_{V^{\diamond}}  \models \bigwedge \neg V^{\triangleright} \mbox{\ \ and\ \ }
I_{\neg V^{\diamond}} \models \bigwedge V^{\triangleright}.
\end{equation}

Let $(I^n)_{n \leq \xi}$ be a terminal well-founded
induction for $D$ with limit $I^\xi = I$. There exists a least ordinal
$n$ such that $P^{I^n} = \Un$ for every $P \in V$ and there
exists at least one $P \in V$ with $P^{I^{n+1}} \not = \Un$. We
distinguish between the case where $P$ is $\Tr$ in $I^{n+1}$ and
the case where $P \in U$ for some non-empty set $U \subseteq \defp{D}$ such that
all atoms of $U$ are $\Fa$ in $I^{n+1}$. We will prove in the
first case that $I_{V^{\diamond}} \models P^{\triangleright}$ and in the second case that
$I_{\neg V^{\diamond}} \models \neg P^{\triangleright}$ for every $P^{\triangleright} \in V^{\triangleright} \cap U^{\triangleright}$. This
contradicts (\ref{PropC}).

\begin{itemize}
\item Assume that $P^{I^n} = \Un$ and $P^{I^{n+1}} = \Tr$.
Then for the rule $P \rul \varphi_{P} \in D$, it holds that $\varphi_{P}^{I^n} = \Tr$.
Consider the corresponding rule $P^{\triangleright} \rul
\varphi_{P}^{\diamond} \in D^{\diamond}$. If we can show that $I_{V^{\diamond}} \models \varphi_{P}^{\diamond}$,
then Lemma \ref{lem:A} will yield that $I_{V^{\diamond}} \models P^{\triangleright}$ which is what
we must prove here.

Consider the $\tau^{\diamond}$-interpretation ${I^n}^{\diamond}$ which extends
$I^n$ by interpreting each symbol $Q^{\triangleright}$ and $Q^{\diamond}$ as
$Q^{I^n}$, i.e., as $\Un$. Clearly, $(\varphi_{P}^{\diamond})^{{I^n}^{\diamond}} =
\varphi_{P}^{I^n} = \Tr$, and it suffices to show that
${I^n}^{\diamond}\leq_p I_{V^{\diamond}}$ to obtain that $I_{V^{\diamond}} \models
\varphi_{P}^{\diamond}$. But this is straightforward since ${I^n}^{\diamond}|_\tau =
I^n \leq_p I = I_{V^{\diamond}}|_\tau$ and $(Q^{\triangleright})^{{I^n}^{\diamond}} =
(Q^{\diamond})^{{I^n}^{\diamond}} = \Un$ for each $Q\in V$.  Hence, it is indeed
the case that ${I^n}^{\diamond} \leq_p I_{V^{\diamond}}$ which leads to the
contradiction.\\ \

\item For the other case, assume that $I^{n+1} = I^n[U/ \Fa]$
where $P\in U$. For each $P\in U\cap V (\neq \emptyset)$ and its
rule $P \rul \varphi_{P} \in D$, it holds that $\varphi_{P}^{I^{n+1}} =
\Fa$. We will use this to show that for each rule $P^{\triangleright}\rul \varphi_{P}^{\diamond}\in
D^{\diamond}$ with $P^{\triangleright}\in U^{\triangleright}\cap V^{\triangleright}$, $\varphi_{P}^{\diamond}$ is false in the
interpretation ${I_{\neg V^{\diamond}}[U^{\triangleright}\cap V^{\triangleright}/ \Fa]}$. Then, since $I_{\neg
V^{\diamond}}$ satisfies $\D^{\diamond}$, we can apply Lemma
\ref{lem:apropertyofwellfoundedmodel} to obtain that each $P^{\triangleright}\in U^{\triangleright}\cap
V^{\triangleright}$ is false in $I_{\neg V^{\diamond}}$. This produces the contradiction with
(\ref{PropC}).

The key point is therefore to show that all these renamed rule bodies
$\varphi_{P}^{\diamond}$ are false in the interpretation ${I_{\neg V^{\diamond}}[U^{\triangleright}\cap
V^{\triangleright}/ \Fa]}$. To prove this, we use the same technique as in the previous
case, namely we construct an interpretation which is less precise than
${I_{\neg V^{\diamond}}[U^{\triangleright}\cap V^{\triangleright}/ \Fa]}$ and which falsifies all the concerned
rule bodies. We choose this interpretation as the
$\tau^{\diamond}$-interpretation $I^{\diamond}$ which extends $I^{n+1}$ by
interpreting each symbol $Q^{\triangleright}$ and $Q^{\diamond}$ as $Q^{I^{n+1}}$, i.e. as
$\Fa$ if $Q\in U\cap V$ and as $\Un$ if $Q\in V\setminus U$. Notice
that for all formulas $\psi$ over $\tau$, it holds that
$\psi^{I^{n+1}} = (\psi^{\diamond})^{I^{\diamond}}$.

Let us verify that $I^{\diamond} \leq_p I_{\neg V^{\diamond}}[U^{\triangleright}\cap
V^{\triangleright}/ \Fa]$. We have $I^{\diamond}|_\tau = I^{n+1} \leq_p I = I_{\neg
V^{\diamond}}[U^{\triangleright}\cap V^{\triangleright}/ \Fa]|_\tau$. The interpretation $I^{\diamond}$ interprets all
symbols $Q^{\diamond}$ as $\Un$ or $\Fa$ whereas $I_{\neg V^{\diamond}}[U^{\triangleright}\cap V^{\triangleright}/ \Fa]$
interprets them as $\Fa$, just like $I_{\neg V^{\diamond}}$. Symbols of $U^{\triangleright}\cap
V^{\triangleright}$ are interpreted as $\Fa$ in both interpretations, and finally, the
remaining symbols of $V^{\triangleright}\setminus U^{\triangleright}$ are interpreted as $\Un$ in
$I^{\diamond}$ which is certainly less precise than in the other
interpretation.

It follows that for every rule $P^{\triangleright}\rul\varphi_{P}^{\diamond}\in D^{\diamond}$ with $P\in
V\cap U$, $\Fa = \varphi_{P}^{I^{n+1}} = (\varphi_{P}^{\diamond})^{I^{\diamond}} \leq_p
(\varphi_{P}^{\diamond})^{I_{\neg V^{\diamond}}[U^{\triangleright}\cap V^{\triangleright}/ \Fa]} = \Fa$.  As explained
before, this leads to the desired contradiction.

\end{itemize}
\end{proof}


\begin{lemma} [Soundness of the definition introduction rule]\label{lem:newdefinitionsoundness}
Let $D$ be a total definition. If $\models D', \Gamma \la \Delta, P' \equiv P$ for every $P \in \defp{D}$, then $\models \Gamma \la \Delta, D$.
\end{lemma}

\begin{proof}
Assume $\models D', \Gamma \la \Delta, P' \equiv P$ for every $P \in \defp{D}$ but $\not \models \Gamma \la \Delta, D$. Then there exists a two-valued $\voc$-interpretation $I$ such that $I \models \bigwedge \Gamma$ but $I \not \models \bigvee \Delta$, $I \not \models D$. Denote by $J$ the two-valued well-founded model of $D$ extending $\res{I}{\openp{D}}$. Because $I \not \models D$, there exists a defined atom $Q$ of $D$ such that $Q^I \not = Q^{J}$. Since $D$ is a total definition and $D'$ is obtained by replacing each occurrence of each defined atom $P$ in $D$ by $P'$, $D'$ is a total definition. Thus, there exists a two-valued $\voc'$-interpretation $I'$ such that $I'$ is the well-founded model of $D'$ extending $I$. Notice that for every $P \in \defp{D}$, $P^{I'} = P^I$. Because neither $\Gamma$ nor $\Delta$ contains an occurrence of an atom $P'$, it holds that $I' \models \bigwedge \Gamma$ and $I' \not \models \bigvee \Delta$. Therefore, by the first assumption, it is obtained that $I' \models P' \equiv P$ for every $P \in \defp{D}$. Also, because $D'$ is obtained by renaming all defined atoms and none of the open atoms, it holds that $P^{J} = (P')^{I'}$ for every $P \in \defp{D}$. Hence, $Q^I = Q^{I'} = (Q')^{I'} = Q^{J}$, a contradiction. Therefore, $\models \Gamma \la \Delta, D$.
\end{proof}

The definition introduction rule is not sound if the inferred definition $D$ is not total. We illustrate it with an example.
\begin{example} \label{ex:deintronontotal}
Consider the definition as follows:
\[ D = \defin{P \rul \neg P}. \]
Let $\Gamma$ and $\Delta$ be empty
sets. It is obvious that $D' = \defin{P' \rul \neg P'}$ is not total. Thus, $\models D' \la P' \equiv P$ but $\not \models \la D$, which shows that the definition introduction rule is not sound when the inferred definition $D$ is non-total.
\end{example}

Notice that all inference rules in \ps except the definition introduction rule are sound with respect to both total and non-total definitions.
By induction on the number of inference rules in a proof of a sequent, we can easily prove the soundness of \ps.

\begin{theorem}[Soundness]\label{theo:soundness}
If a sequent $\Gamma \la \Delta$ is provable in \ps without using the definition introduction rule, then $\models \Gamma \la \Delta$. If a sequent $\Gamma \la \Delta$ is provable in \ps and all definitions occurring in $\Gamma$ and $\Delta$ are total, then $\models \Gamma \la \Delta$.
\end{theorem}

\subsection{Completeness}
\ps is not complete in general. Intuitively, this is because the only inference rules that allow to introduce a positive occurrence of a definition in the succedent of a sequent are the axioms, the weakening rules and the definition introduction rule. As shown in the above subsection, the definition introduction rule is not sound with respect to non-total definitions. Thus, no other inference rule allows to derive a non-total definition from some propositional formulas. Therefore, one cannot synthesize non-total definitions with \ps, i.e., not all valid sequents of the form $\Gamma \la D$, where $D$ is a non-total definition, can be proved in this system.

We will however prove the completeness for a restricted class of sequents, namely the sequents $\Gamma \la \Delta$ such that every definition occurring negatively in $\Gamma$ or positively in $\Delta$ must be total.
The main difficulty in the completeness proof for \ps is to handle the definitions in the sequents (We already know that the propositional part of \ps is complete. See e.g. \cite{Szabo69,Takeuti75}).

First, we focus on the completeness of sequents of the form $D, \Gamma \la \Delta$, where $\Gamma$ and $\Delta$ are
sets of PC-formulas and $D$ is a definition. Notice that the definition $D$ appearing in the sequent may be non-total.

\begin{lemma}\label{lem:gencomp}
Let $D$ be a definition and $\Gamma$ a
set of open literals of $D$ such that for every $Q \in \openp{D}$ either $Q \in \Gamma$ or $\neg Q \in \Gamma$. Let $I_O$ be the unique two-valued $\openp{D}$-interpretation such that $I_O \models \bigwedge \Gamma$ and $I$ the well-founded model of $D$ extending $I_O$. If $L$ is a defined literal of $D$ such that $L^I = \Tr$, then $D, \Gamma \la L$ is provable in \ps.
\end{lemma}

\begin{proof}
Let $(I^n)_{n \leq \xi}$ be a terminal well-founded induction for $D$ extending $I_O$ with the limit $I^\xi = I$. Denote by $\Delta^n$ a
set of all defined literals $L$ such that $L^{I^n} = \Tr$ in arbitrary order. We prove that $\Delta^n, D, \Gamma \la L$ is provable in \ps for all $L \in \Delta^{n +1} \setminus \Delta^n$. For each $L \in \Delta^{n+1} \setminus \Delta^{n}$, $L^{I^n} = \Un$ and $L^{I^{n+1}}= \Tr$. We distinguish between the case where $\Delta^{n+1} \setminus \Delta^{n}$ contains one positive literal and the case where it contains a set of negative literals.
\begin{itemize}
\item Assume that $\Delta^{n+1} \setminus \Delta^n$ consists of one defined atom $P$.
For every two-valued $\voc$-interpretation $J$ such that $J$ is a model of $\bigwedge \Gamma$ and $\bigwedge \Delta^n$, $I^n \leq_p J$. Indeed, $L^{I^n} = L^J = \Tr$ for every $L \in \Gamma$, $L^{I^n} = L^J = \Tr$ for every $L \in \Delta^n$ and for every other atom $Q \in \voc$, $Q^{I^n} = \Un \leq_p Q^J$. $P^{I^{n+1}} = \Tr$, hence $\varphi_{P}^{I^n} = \Tr$.
It follows that $\varphi_{P}^{J} = \Tr$. Thus, $\models \Delta^n, \Gamma \la \varphi_{P}$. Therefore, by the completeness of the propositional part of \ps, the sequent $\Delta^n, \Gamma \la \varphi_{P}$ is provable in \ps. Hence, by
the right definition rule,
$\Delta^{n}, D, \Gamma \la P$ is provable in \ps.
\item For the other case, assume that $\Delta^{n+1} \setminus \Delta^n$ is a
set of negative literals. Denote the set $\{ P \mid \neg P \in \Delta^{n+1} \setminus \Delta^n \}$ by $U$.
Recall that $I^{n+1} = I^n[U/ \Fa]$. $P^{I^{n+1}} = \Fa$ for each $P \in U$, hence $\varphi_{P}^{I^{n+1}} = \Fa$. Consider the interpretation ${I^{n+1}}^{\triangleright} = I^n[U^{\triangleright} / \Fa]$. There are two simple observations that can be made about ${I^{n+1}}^{\triangleright}$ and each $\varphi_{P}^{\triangleright}$:
\begin{itemize}
\item ${I^{n+1}}^{\triangleright} \leq_p J'$ for every two-valued $\voc \cup U^{\triangleright}$-interpretation $J'$ such that $J'$ satisfies $\bigwedge \Gamma$, $\bigwedge \Delta^n$ and $\bigwedge \neg U^{\triangleright}$: indeed, $L^{{I^{n+1}}^{\triangleright}} = L^{J'} = \Tr$ for every $L \in \Gamma$, $L^{{I^{n+1}}^{\triangleright}} = L^{J'} = \Tr$ for every $L \in \Delta^n$, $(P^{\triangleright})^{{I^{n+1}}^{\triangleright}} = {P^{\triangleright}}^{J'} = \Fa$ for every $P^{\triangleright} \in U^{\triangleright}$ and
$Q^{{I^{n+1}}^{\triangleright}} = \Un \leq_p Q^{J'}$ for every other atom $Q \in \voc \cup U^{\triangleright}$.
\item $(\varphi_{P}^{\triangleright})^{{I^{n+1}}^{\triangleright}} = \varphi_{P}^{I^{n+1}} = \Fa$: obvious from the construction of ${I^{n+1}}^{\triangleright}$ and $\varphi_{P}^{\triangleright}$.
\end{itemize}
Combining these results, we obtain $(\varphi_{P}^{\triangleright})^{J'} = \Fa$ for every two-valued interpretation $J'$ satisfying $\bigwedge \Gamma$, $\bigwedge \Delta^n$ and $\bigwedge \neg U^{\triangleright}$. It follows that $\models \neg U^{\triangleright}, \Delta^n, \Gamma \la \neg \varphi_{P}^{\triangleright}$ for every $P \in U$. By the completeness of the propositional part of \ps, the left
   definition rule and the right $\neg$ rule 
   the sequent $\Delta^{n}, D, \Gamma \la \neg P$ is provable in \ps for every
   $P \in U$.
\end{itemize}
Since $(I^n)_{n \leq \xi}$ is a terminal well-founded induction for $D$ with the limit $I = I^\xi$, it is obvious that the set of defined literals $L$ for which $L^I = \Tr$ is exactly the set of all defined literals in $\Delta^\xi$.  Thus, by using the cut rule, it is easy to show by induction on $n$ that if $L$ is a defined literal of $D$ such that $L^I = \Tr$, the sequent $D, \Gamma \la L$ is provable in \ps.
\end{proof}

Notice that in the above lemma, we do not require the totality of the definition. So the definition $D$ can be non-total and the well-founded model of $D$ may be a three-valued interpretation.

\begin{lemma}\label{lem:defcomp}
Let $D$ be a total definition and let $\Gamma$ be a
set of open literals of $D$, such that for every atom $Q \in \openp{D}$ either $Q \in \Gamma$ or $\neg Q \in \Gamma$. Let $L$ be a defined literal of $D$. If $\models D, \Gamma \la L$, then $D, \Gamma \la L$ is provable in \ps.
\end{lemma}

\begin{proof}
Assume that $\models D, \Gamma \la L$. Let $I_O$ be the unique two-valued $\openp{D}$-interpretation such that $I_O \models \bigwedge \Gamma$. Because $D$ is total, $I_O$ can be extended to a two-valued well-founded model $I$ of $D$ such that $I \models \bigwedge \Gamma$ and $I \models D$. Then since $\models D, \Gamma \la L$, it holds that $L^I = \Tr$. Thus, by Lemma~\ref{lem:gencomp}, $D, \Gamma \la L$ is provable in \ps.
\end{proof}

\begin{lemma}\label{lem:defarcomp}
Let $D$ be a total definition and $\Gamma$ an arbitrary consistent 
set of literals. If $L$ is a defined literal of $D$ such that $\models D, \Gamma \la L$, then $D, \Gamma \la L$ is provable in \ps.
\end{lemma}

\begin{proof}
Let $\Gamma'$ be an arbitrary extension of $\Gamma$ such that for every open
atom $Q$ of $D$, either $Q \in \Gamma'$ or $\neg Q \in \Gamma'$.
First, we want to show that $D, \Gamma' \la L$ is provable in \ps. It holds that $\models D, \Gamma' \la L$ because
$\models D, \Gamma \la L$.
Consider the
set $\Gamma''$ of
all open literals of $D$ in $\Gamma'$. If $\models D, \Gamma'' \la L$, then by the
previous lemma, $D, \Gamma'' \la L$ is provable in \ps, and by the left
weakening rule,
so is $D, \Gamma' \la L$.  If $\not \models D, \Gamma'' \la
L$, then by totality of $D$, $ \models D, \Gamma'' \la \neg L$ and
hence, $ \models D, \Gamma' \la \neg L$. This means that $D \land \bigwedge \Gamma'$ is unsatisfiable, which implies that for some defined literal $L'$
in $\Gamma'$, $\models D, \Gamma'' \la \neg L'$. By the previous lemma and
the left weakening rule,
$D, \Gamma' \la \neg L'$ is provable in \ps. It is obvious that
$D, \Gamma' \la L'$ is an axiom because $L'$ is a literal in $\Gamma'$. Then
we can use the left $\neg$ rule, the cut rule and the right weakening rule
to show that $D, \Gamma' \la L$ is provable in \ps.

Given that the sequents $D, \Gamma' \la L$ are provable in \ps for all extensions $\Gamma'$ of $\Gamma$, by using
the right $\neg$ rule and
the cut
rule on all $D, \Gamma' \la L$, an \ps-proof for $D, \Gamma
\la L$ can be constructed.
\end{proof}

\begin{lemma} \label{lem:nontotalcomp}
Let $D$ be a definition and $\Gamma$ a 
set of open literals of $D$, such that for every atom $Q \in \openp{D}$ either $Q \in \Gamma$ or $\neg Q \in \Gamma$. If $\models D, \Gamma \la \bot$, then $D, \Gamma \la \bot$ is provable in \ps.
\end{lemma}

\begin{proof}
Let $I_O$ be the unique two-valued $\openp{D}$-interpretation such that $I_O \models \bigwedge \Gamma$ and $(I^n)_{n \leq \xi}$ a terminal well-founded induction for $D$ extending $I_O$ with limit $I^\xi = I$. Because $\models D, \Gamma \la \bot$, there is no two-valued well-founded model for $D$ extending $I_O$. Hence $I$ is a three-valued $\voc$-interpretation. Denote by $E$ the set of all defined atoms of $D$ which are not unknown in $I$ and $V$ the set $\defp{D} \setminus E$.
For each $P \in E$, we define a literal $L_P$ as follows:
\[ L_P = \left\{
\begin{array}{ll}
P & \mbox{\quad if $P^I = \Tr$} \\
\neg P & \mbox{\quad if $P^I = \Fa$}\\
\end{array}  \right. .\]
Denote by $K$ the set $\{ L_P \mid P \in E \}$ of literals. We first want to show that
\begin{equation} \label{CooC}
\models D^{\diamond}, V^{\diamond}, K, \Gamma \la \bigwedge \neg V^{\triangleright} \mbox{\ and \ }
\models D^{\diamond}, \neg V^{\diamond}, K, \Gamma \la \bigwedge V^{\triangleright}.
\end{equation}
Consider the vocabulary $\tau^{\diamond} =
\tau\cup V^{\triangleright}\cup V^{\diamond}$. $I$ can be expanded into two
$\tau^{\diamond}$-interpretations $I_{V^{\diamond}}$ and $I_{\neg V^{\diamond}}$ as follows:
$$I_{V^{\diamond}} = (I[V^{\diamond}/ \Tr])^{D^{\diamond}} \mbox{\ \ and \ \ } I_{\neg V^{\diamond}} =
(I[V^{\diamond}/ \Fa])^{D^{\diamond}}.$$
Since $D^{\diamond}$ is a positive definition, hence total definition with open symbols
$\tau\cup V^{\diamond}$, both interpretations are well-defined. Moreover it is obvious that $I_{V^{\diamond}}$, respectively $I_{\neg V^{\diamond}}$, is the only interpretation satisfying:
$$I_{V^{\diamond}} \models D^{\diamond} \land \bigwedge V^{\diamond} \land \bigwedge K \land \bigwedge \Gamma,
\mbox{\ respectively \ }I_{\neg V^{\diamond}} \models D^{\diamond} \land
\bigwedge \neg V^{\diamond} \land \bigwedge K \land \bigwedge \Gamma.$$

In order to prove (\ref{CooC}), it suffices to show that
\begin{equation} \label{PropB}
I_{V^{\diamond}}  \models \bigwedge \neg V^{\triangleright} \mbox{\ \ and\ \ }
I_{\neg V^{\diamond}} \models \bigwedge V^{\triangleright}.
\end{equation}
\begin{itemize}
\item We want to prove that $I_{V^{\diamond}} \models \bigwedge \neg V^{\triangleright}$. For any $P \in V$ with its rule $P \rul \varphi_{P} \in D$, $P^{\triangleright} \rul \varphi_{P}^{\diamond}$ is the corresponding rule for $P^{\triangleright}$ in $D^{\diamond}$. If we can show that $I_{V^{\diamond}}[V^{\triangleright}/ \Fa] \models \neg \varphi_{P}^{\diamond}$ for each $P^{\triangleright} \in V^{\triangleright}$ with its rule $P^{\triangleright} \rul \varphi_{P}^{\diamond}$, then since $I_{V^{\diamond}}$ satisfies $D^{\diamond}$, we can apply Lemma~\ref{lem:apropertyofwellfoundedmodel} to obtain that each $P^{\triangleright} \in V^{\triangleright}$ is false in $I_{V^{\diamond}}$, which is what we must prove here.

Consider the $\voc^{\diamond}$-interpretation $I^{\diamond}$ which extends $I$ by interpreting each symbol $Q^{\triangleright}$ and $Q^{\diamond}$ as $Q^I$ for each $Q \in V$, i.e., as $\Un$. Clearly, for every $P \in V$ with its rule $P \rul \varphi_{P} \in D$, $(\varphi_{P}^{\diamond})^{I^{\diamond}} = \varphi_{P}^I = \Un$, and since $(\varphi_{P}^{\diamond})^{I_{V^{\diamond}}[V^{\triangleright}/ \Fa]} \not = \Un$, it is sufficient to show that $(\varphi_{P}^{\diamond})^{I_{V^{\diamond}}[V^{\triangleright}/ \Fa]} \leq (\varphi_{P}^{\diamond})^{I^{\diamond}}$ to obtain that $I_{V^{\diamond}}[V^{\triangleright}/ \Fa] \models \neg \varphi_{P}^{\diamond}$ for every $P^{\triangleright} \in V^{\triangleright}$ with its rule $P^{\triangleright} \rul \varphi_{P}^{\diamond} \in D^{\diamond}$. This can be verified by the following observations.
\begin{itemize}
\item $\res{I_{V^{\diamond}}[V^{\triangleright} / \Fa]}{\voc} = \res{I^{\diamond}}{\voc}$.
\item For every $Q \in V$, every occurrence of $Q^{\triangleright}$ in $\varphi_{P}^{\diamond}$ is positive and $(Q^{\triangleright})^{I_{V^{\diamond}}[V^{\triangleright} / \Fa]} = \Fa \leq (Q^{\triangleright})^{I^{\diamond}} = \Un$.
\item For every $Q \in V$, every occurrence of $Q^{\diamond}$ in $\varphi_{P}^{\diamond}$ is negative and $(Q^{\diamond})^{I_{V^{\diamond}}[V^{\triangleright} / \Fa]} = \Tr \geq (Q^{\diamond})^{I^{\diamond}} = \Un$
\end{itemize}
Hence, it is indeed the case that $(\varphi_{P}^{\diamond})^{I_{V^{\diamond}}[V^{\triangleright}/ \Fa]} \leq (\varphi_{P}^{\diamond})^{I^{\diamond}}$, as desired.
\item We want to prove that $I_{\neg V^{\diamond}} \models \bigwedge V^{\triangleright}$. Assume toward contradiction that there exists a non-empty set $F^{\triangleright} \subseteq V^{\triangleright}$
such that $I_{\neg V^{\diamond}} \models \bigwedge \neg F^{\triangleright}$ and for the set $T^{\triangleright} = V^{\triangleright} \setminus F^{\triangleright}$, $I_{\neg V^{\diamond}} \models \bigwedge T^{\triangleright}$.
Consider the $\voc$-interpretation $I^1 = I[F/ \Fa]$. If we can show that $\varphi_{P}^{I^1} = \Fa$ for every $P \in F$ with its rule $P \rul \varphi_{P} \in D$, then since for each $P \in F$ and its rule $P \rul \varphi_{P} \in D$, $P^I = \Un$ and $\varphi_{P}^{I^1} = \Fa$,
$I$ can be extended to $I^1$ in the well-founded induction $(I^n)_{n \leq \xi}$ for $D$.
This produces the contradiction to that $I$ is the limit of $(I^n)_{n \leq \xi}$. To prove that $\varphi_{P}^{I^1} = \Fa$ for every $P \in F$ with the rule $P \rul \varphi_{P} \in D$, we first choose a $\voc^{\diamond}$-interpretation $I^{\diamond}$ which extends $I^1$ by interpreting each symbol $Q^{\triangleright}$ and $Q^{\diamond}$ as $Q^{I^1}$, i.e., as $\Fa$ if $Q \in F$ and as $\Un$ if $Q \in T$. Notice that for all formulas $\psi$ over $\voc$, it holds that $\psi^{I^1} = (\psi^{\diamond})^{I^{\diamond}}$. Thus, it is sufficient to show that $(\varphi_{P}^{\diamond})^{I^{\diamond}} = \Fa$ for every $P^{\triangleright} \in F^{\triangleright}$ with the rule $P^{\triangleright} \rul \varphi_{P}^{\diamond} \in D^{\diamond}$. Since $I_{\neg V^{\diamond}} \models \neg P^{\triangleright}$ for each $P^{\triangleright} \in F^{\triangleright}$ and $I_{\neg V^{\diamond}}$ is a model of $D^{\diamond}$, by Lemma~\ref{lem:A}, we have that $(\varphi_{P}^{\diamond})^{I_{\neg V^{\diamond}}} = \Fa$ for every $P^{\triangleright} \in F^{\triangleright}$ with the rule $P^{\triangleright} \rul \varphi_{P}^{\diamond} \in D^{\diamond}$. If we can have that $(\varphi_{P}^{\diamond})^{I^{\diamond}} \leq (\varphi_{P}^{\diamond})^{I_{\neg V^{\diamond}}} = \Fa$, it holds that $(\varphi_{P}^{\diamond})^{I^{\diamond}} = \Fa$, which is exactly what we need.

We can verify that $(\varphi_{P}^{\diamond})^{I^{\diamond}} \leq (\varphi_{P}^{\diamond})^{I_{\neg V^{\diamond}}}$ by the following facts.
\begin{itemize}
\item $\res{I_{\neg V^{\diamond}}}{\voc} = \res{I^{\diamond}}{\voc}$.
\item Every occurrence of $Q^{\triangleright}$ in $\varphi_{P}^{\diamond}$ is positive and $(Q^{\triangleright})^{I_{\neg V^{\diamond}}} = (Q^{\triangleright})^{I^{\diamond}} = \Fa$ for each $Q^{\triangleright} \in F^{\triangleright}$ while $(Q^{\triangleright})^{I^{\diamond}} = \Un \leq (Q^{\triangleright})^{I_{\neg V^{\diamond}}} = \Tr$ for each $Q^{\triangleright} \in V^{\triangleright} \setminus F^{\triangleright}$.
\item Every occurrence of $Q^{\diamond}$ in $\varphi_{P}^{\diamond}$ is negative and $(Q^{\diamond})^{I_{\neg V^{\diamond}}} = (Q^{\diamond})^{I^{\diamond}} = \Fa$ for each $Q^{\diamond} \in F^{\diamond}$ while $(Q^{\diamond})^{I^{\diamond}} = \Un \geq (Q^{\diamond})^{I_{\neg V^{\diamond}}} = \Fa$ for each $Q^{\diamond} \in V^{\diamond} \setminus F^{\diamond}$.
\end{itemize}
Hence, it is the case that $(\varphi_{P}^{\diamond})^{I^{\diamond}} \leq (\varphi_{P}^{\diamond})^{I_{\neg V^{\diamond}}} = \Fa$, as desired.
\end{itemize}

Therefore, it is obtained that $\models D^{\diamond}, V^{\diamond}, K, \Gamma \la \bigwedge \neg V^{\triangleright}$ and $\models D^{\diamond}, \neg V^{\diamond}, K, \Gamma \la \bigwedge V^{\triangleright}$. $D^{\diamond}$ is a total definition, hence by using Lemma~\ref{lem:defarcomp} 
and the right $\land$ rule, both $V^{\diamond}, D^{\diamond}, K, \Gamma \la \bigwedge \neg V^{\triangleright}$ and $\neg V^{\diamond}, D^{\diamond}, K, \Gamma \la \bigwedge V^{\triangleright}$ are provable in \ps. It follows from the non-total definition rule that 
$K, D, \Gamma \la \bot$ is provable in \ps. 
Since $I$ is a well-founded model of $D$ extending $I_O$ and $L^I = \Tr$ for each $L \in K$, using Lemma~\ref{lem:gencomp}, it holds that for each $L \in K$, $D, \Gamma \la L$ is provable in \ps. Consequently, by the multiple use of the cut rule on $K, D, \Gamma \la \bot$ and $D, \Gamma \la L$ for each $L \in K$,  $D, \Gamma \la \bot$ is provable in \ps.
\end{proof}

\begin{lemma} \label{lem:defliteralcomp}
Let $D$ be a definition and $\Gamma$ a 
set of open literals of $D$ such that for every atom $Q \in \openp{D}$, either $Q \in \Gamma$ or $\neg Q \in \Gamma$. Let $L$ be a defined literal of $D$. If $\models D, \Gamma \la L$, then $D, \Gamma \la L$ is provable in \ps.
\end{lemma}

\begin{proof}
Assume $\models D, \Gamma \la L$. Let $I_O$ be the unique two-valued $\openp{D}$-interpretation such that $I_O \models \bigwedge \Gamma$.
If $\not \models D, \Gamma \la \bot$, then $I_O$ can be extended to the two-valued well-founded model $I$ of $D$ such that $I \models \bigwedge \Gamma$ and $I \models D$. Since $\models D, \Gamma \la L$, it holds that $I \models L$. Thus, by Lemma~\ref{lem:gencomp}, $D, \Gamma \la L$ is provable in \ps. If $\models D, \Gamma \la \bot$, then by Lemma~\ref{lem:nontotalcomp}, $D, \Gamma \la \bot$ is provable in \ps. Hence, by the right weakening rule, $D, \Gamma \la L$ is provable in \ps.
\end{proof}

\begin{lemma} \label{lem:falsecomp}
Let $D$ be a definition and $\Gamma$ an arbitrary consistent 
set of literals. If $\models D, \Gamma \la \bot$, then $D, \Gamma \la \bot$ is provable in \ps.
\end{lemma}
To prove this, we use the same technique as in the proof of Lemma~\ref{lem:defarcomp}. We omit the details of the proof here.

\begin{lemma} \label{lem:arbitraryliteralcomp}
Let $D$ be a definition, $\Gamma$ an arbitrary consistent set of literals and $L$ a defined literal of $D$. If $\models D, \Gamma \la L$, then $D, \Gamma \la L$ is provable in \ps.
\end{lemma}

\begin{proof}
If $\Gamma \cup \{ \neg L \}$ is an inconsistent set of literals, we have that $D, \Gamma \la L$ is an axiom and thus, $D, \Gamma \la L$ is provable in \ps. If $\Gamma \cup \{ \neg L \}$ is consistent, because $\models D, \Gamma, \neg L \la \bot$, by the previous lemma, it is obtained that $D, \Gamma, \neg L \la \bot$ is provable in \ps. Then by the $\neg$ rules and the cut rule, we can conclude that $D, \Gamma \la L$ is provable in \ps.
\end{proof}

The remainder of the completeness proof for the class of sequents, namely the sequents $\Gamma, D \la \Delta$ where $\Gamma$ and $\Delta$ are sets of PC-formulas and $D$ is a definition, will use a standard technique: we construct the so called {\em reduction tree} for a sequent $\Gamma \la \Delta$. We follow the approach from~\cite{Takeuti75}. First, we introduce some terminology.

\begin{definition} \label{def:reductiontree}
A {\em reduction tree} for a sequent $S = \Gamma \la \Delta$ is a tree $T_S$ of sequents. The root of $T_S$ is $S$. Moreover, $T_S$ is constructed by applying one of the following reductions on each non-leaf $\Pi \la \Sigma$.
\begin{itemize}
\item [--] (left $\neg$ reduction) $\Pi$ contains a sequent formula $\neg A$, then write down $\Pi \setminus \{ \neg A \} \la \Sigma, A$ as the unique child of $\Pi \la \Sigma$.
\item [--] (right $\neg$ reduction) $\Sigma$ contains a sequent formula $\neg A$, then write down $A, \Pi \la \Sigma \setminus \{\neg A\}$ as the unique child of $\Pi \la \Sigma$.
\item [--] (left $\land$ reduction) $\Pi$ contains a sequent formula $A \land B$, then write down $A, B, \Pi \setminus \{ A \land B \} \la \Sigma$ as the unique child of $\Pi \la \Sigma$.
\item [--] (right $\land$ reduction) $\Sigma$ contains a sequent formula $A \land B$, then write down $\Pi \la \Sigma \setminus \{A \land B \}, A$ and $\Pi \la \Sigma \setminus \{A \land B \}, B$ as two children of $\Pi \la \Sigma$.
\item [--] (left $\lor$ reduction) $\Pi$ contains a sequent formula $A \lor B$, then write down $A, \Pi \setminus \{ A \lor B \} \la \Sigma$ and $B, \Pi \setminus \{A \lor B \} \la \Sigma$ as two children of $\Pi \la \Sigma$.
\item [--] (right $\lor$ reduction) $\Sigma$ contains a sequent formula $A \lor B$, then write down $\Pi \la \Sigma \setminus \{ A \lor B \}, A, B$ as the unique child of $\Pi \la \Sigma$.
\item [--] (definition introduction reduction) $\Sigma$ contains a sequent formula $D$, which is a total definition with $\defp{D} = \{ P_1, \ldots, P_n \}$, then write down $D', \Pi \la \Sigma \setminus \{ D \}, P'_{i} \equiv P_i$ for each $i \in [1,n]$ as $n$ children of $\Pi \la \Sigma$.
\end{itemize}
In addition, each leaf of $T_S$ is either an axiom, or none of the above reductions is possible.
\end{definition}

Observe that the definition introduction reduction corresponds to the definition introduction rule while each other reduction respectively corresponds to a logical inference rule. Each leaf node of a reduction tree is either an axiom or a sequent of the form $D_1, \ldots, D_n, \Gamma \la \Delta$ where $\Gamma$ and $\Delta$ are sets of atoms with $\Gamma \cap \Delta = \emptyset$ and $D_1, \ldots, D_n$ are definitions.

\begin{definition}\label{def:counterpre}
An inference rule {\em preserves counter-model} if for each instance of the inference rule, a counter-model for one of the premises of the instance is the same as a counter-model for the conclusion of the instance.
\end{definition}
The following property can easily be verified.
\begin{proposition} \label{prop:logicalcounterpre}
All the logical inference rules preserve counter-models.
\end{proposition}

\begin{lemma} \label{lem:newdefinitioncounterpre}
The definition introduction rule preserves counter-model.
\end{lemma}

\begin{proof}
Let $D$ be a total definition. Then $D'$ is a total definition because of its construction. Assume that $I$ is a counter-model of $D', \Gamma \la \Delta, P' \equiv P$ for
some $P \in \defp{D}$, but $I$ is not a counter-model of $\Gamma \la \Delta, D$. Since $D$ and $D'$ are total,
$I$ is a two-valued interpretation satisfying $D'$, $\bigwedge \Gamma$, $ \neg \bigvee \Delta$ and $\neg (P' \equiv P)$.
Because $I$ is not a counter-model for $\Gamma \la \Delta, D$, it holds that $I \models D$. Obviously from the construction of $D'$ and the
fact that $I$ satisfies both $D$ and $D'$, we conclude that $I \models P' \equiv P$ for every $P \in \defp{D}$, a contradiction.
\end{proof}

Then we obtain the property of reduction trees as follows.

\begin{proposition}\label{prop:reductiontree}
For each sequent $S =
  \Gamma \la \Delta$, \begin{inlinenum}
  \item\label{exred} there exists a reduction tree $T_S$,
  \item\label{prleaf} if
  all leaf nodes of a reduction tree $T_S$ are provable in \ps,
  then the root sequent is provable in \ps, and
  \item\label{countleaf}, there exists a leaf node of $T_S$ such that a counter-model for this leaf node is a counter-model for the root.
  \end{inlinenum}
\end{proposition}

\begin{proof}
  Clearly, a reduction tree exists because it can be constructed by a
  non-deterministic reduction process.
  Because each reduction in a reduction tree corresponds to either the definition introduction rule or a logical inference rule,
  by using the corresponding inference rule, it is easy to prove
  that if the children of a node in a reduction tree are provable in \ps,
  then the node itself is provable in \ps. Therefore, the root sequent is provable in \ps if all leaf nodes of the reduction tree are provable in \ps.

  A counter-model for a leaf is a counter-model for the root because all
  the logical inference rules 
  and the definition introduction rule preserve counter-models
  by Proposition~\ref{prop:logicalcounterpre} and Lemma~\ref{lem:newdefinitioncounterpre} and each non-leaf node can be proved
  from its children using only those inference rules.
\end{proof}

We are now ready to prove the completeness theorem of the sequents of the form $D, \Gamma \la \Delta$, where $\Gamma$ and $\Delta$ are sets of PC-formulas and $D$ is a definition.

\begin{theorem}[Completeness for one definition in the antecedent] \label{theo:completenessforonedefinition}
Let $\Gamma$ and $\Delta$ be
sets of PC-formulas and $D$ a definition. If $\models D, \Gamma \la \Delta$, then $D, \Gamma \la \Delta$ is provable in \ps.
\end{theorem}

\begin{proof}
First, a reduction tree is constructed from the root $D, \Gamma \la \Delta$. Every leaf of the reduction tree must be an axiom or a sequent of the form $D, \Pi \la \Sigma$, where $\Pi$ and $\Sigma$ are (possibly empty) sets of propositional atoms satisfying that \begin{inlinenum}
  \item\label{(a-i)} $\Pi$ and $\Sigma$ have no atom in common, and
  \item\label{(a-ii)} when $\Sigma$ is not empty, $\Pi$ or $\Sigma$ contains at least one defined atom of $D$.
  \end{inlinenum}
By \ref{countleaf} of Proposition~\ref{prop:reductiontree}, if $\models D, \Gamma \la \Delta$, then $\models D, \Pi \la \Sigma$. Hence, if $\Sigma$ is empty, by Lemma~\ref{lem:falsecomp}, it is obtained that $D, \Pi \la \Sigma$ is provable in \ps. If $\Sigma$ is not empty, by Lemma~\ref{lem:arbitraryliteralcomp}, the $\neg$ rules and the weakening rules, $D, \Pi \la \Sigma$ is provable in \ps. Extending for every leaf $D, \Pi \la \Sigma$ the branch that ends in that leaf with the prooftree for that leaf, yields an \ps-proof for $D, \Gamma \la \Delta$.
\end{proof}

\ps remains complete for sequents of the form $D_1, \ldots, D_n, \Gamma \la \Delta$, where $\Gamma$ and $\Delta$ are
sets of PC-formulas and multiple definitions are allowed in the antecedent.

\begin{lemma}
Let $D_1, \ldots, D_n$ be definitions and $\Gamma$ an arbitrary consistent 
set of literals. If $\models D_1, \ldots, D_n, \Gamma \la \bot$, then $D_1, \ldots, D_n, \Gamma \la \bot$ is provable in \ps.
\end{lemma}

\begin{proof}
Let $\Gamma'$ be an arbitrary extension of $\Gamma$ such that for every $D_i \in \{D_1, \ldots, D_n \}$ and every open atom $Q$ of $D_i$, either $Q \in \Gamma'$ or $\neg Q \in \Gamma'$. First, we want to show that $D_1, \ldots, D_n, \Gamma' \la \bot$ is provable in \ps. It holds that $\models D_1, \ldots, D_n, \Gamma' \la \bot$ because $\models D_1, \ldots, D_n, \Gamma \la \bot$. Consider the
set $\Gamma''$ of all open literals of all definitions $D_1, \ldots, D_n$ in $\Gamma'$. We distinguish between the case where $\models D_1, \ldots, D_n, \Gamma'' \la \bot$ and the case where $\not \models D_1, \ldots, D_n, \Gamma'' \la \bot$.
\begin{itemize}
\item In the first case where $\models D_1, \ldots, D_n, \Gamma'' \la \bot$, we distinguish between the subcase where there exists at least one $D_i \in \{D_1, \ldots, D_n \}$ such that $\models D_i, \Gamma'' \la \bot$ and the subcase where for every $D_i \in \{D_1, \ldots, D_n \}$ it holds that $\not \models D_i, \Gamma'' \la \bot$.
\begin{itemize}
\item In the first subcase, $\models D_i, \Gamma'' \la \bot$, hence by Lemma~\ref{lem:falsecomp}, $D_i, \Gamma'' \la \bot$ is provable in \ps. Then by using the left weakening rule,
we conclude that $D_1, \ldots, D_n, \Gamma' \la \bot$ is provable in \ps.

\item In the other subcase, it holds that $\not \models D_i, \Gamma'' \la \bot$ for every $D_i \in \{D_1, \ldots, D_n \}$. Thus, for every $D_i \in \{D_1, \ldots, D_n \}$, there exists a unique two-valued well-founded model $I_i$ of $D_i$ such that $I_i \models D_i$ and $I_i \models \bigwedge \Gamma''$. Because $D_1 \land \ldots \land D_n \land \bigwedge \Gamma''$ is unsatisfiable, for some $I_i$ and $I_j$ such that $i \not = j$
and for some defined literal $L$, it can be implied that $I_i \models L$ and $I_j \models \neg L$. Thus, we have that $\models D_i, \Gamma'' \la L$ and $\models D_j, \Gamma'' \la \neg L$. Therefore, by Lemma~\ref{lem:arbitraryliteralcomp}, it is concluded that both $D_i, \Gamma'' \la L$ and $D_j, \Gamma'' \la \neg L$ are provable in \ps. Then we can use the left weakening rule, the left $\neg$ rule and the cut rule
to show that $D_1, \ldots, D_n, \Gamma' \la \bot$ is provable in \ps.
\end{itemize}
\item In the other case where $\not \models D_1, \ldots, D_n, \Gamma'' \la \bot$, hence there exists a unique two-valued interpretation $I$ such that $I \models D_1 \land \ldots \land D_n \land \bigwedge \Gamma''$. Because $\not \models D_1, \ldots, D_n, \Gamma'' \la \bot$, for each $D_i \in \{D_1, \ldots, D_n \}$, it holds that $\not \models D_i, \Gamma'' \la \bot$ and hence, there exists a unique two-valued well-founded model $I_i$ of $D_i$ such that $I_i \models D_i$ and $I_i \models \bigwedge \Gamma''$. Therefore, for each $D_i$ and each defined atom $P \in \defp{D_i}$, $P^{I_i} = P^I$. Since $D_1 \land \ldots \land D_n \land \bigwedge \Gamma''$ is satisfiable but $D_1 \land \ldots \land D_n \land \bigwedge \Gamma'$ is unsatisfiable, it can be implied that for some defined literal $L'$ in $\Gamma'$, $\models D_1, \ldots, D_n, \Gamma'' \la \neg L'$. Assume that $L'$ is a defined literal of $D_i$. Because ${L'}^{I_i} = {L'}^I = \Fa$, we have that $\models D_i, \Gamma'' \la \neg L'$. By Lemma~\ref{lem:arbitraryliteralcomp} and the left weakening rule,
$D_i, \Gamma' \la \neg L'$ is provable in \ps. It is obvious that $D_i, \Gamma' \la L'$ is an axiom because $L'$ is a literal in $\Gamma'$. Then we can use the left weakening rule, the left $\neg$ rule and the cut rule
to show that $ D_1, \ldots, D_n, \Gamma' \la \bot$ is provable in \ps.
\end{itemize}
Given that the sequents $D_1, \ldots, D_n, \Gamma' \la \bot$ are provable in \ps for all extensions $\Gamma'$ of $\Gamma$, by using
the right $\neg$ rule and the cut rule on all $D_1, \ldots, D_n, \Gamma' \la \bot$, we can construct an \ps-proof for $D_1, \ldots, D_n, \Gamma \la \bot$.
\end{proof}

\begin{lemma}\label{lem:comofmulti}
Let $D_1, \ldots, D_n$ be definitions and let $\Gamma$ and $\Delta$ be
sets of atoms. If $\models D_1, \ldots, D_n, \Gamma \la \Delta$, then $D_1, \ldots, D_n, \Gamma \la \Delta$ is provable in \ps.
\end{lemma}

\begin{proof}
The proof is trivial if $D_1, \ldots, D_n, \Gamma \la \Delta$ is an axiom, hence we assume that $D_1, \ldots, D_n, \Gamma \la \Delta$ is not an axiom,
i.e. $\Gamma \cap \Delta = \emptyset$. Because $\Gamma, \neg \Delta$ is a consistent
set of literals and $\models D_1, \ldots, D_n, \Gamma, \neg \Delta \la \bot$, by the previous lemma, we have that $D_1, \ldots, D_n, \Gamma, \neg \Delta \la \bot$ is provable in \ps. Then by
the $\neg$ rules and the cut rule, we can conclude that $ D_1, \ldots, D_n, \Gamma \la \Delta$ is provable in \ps.
\end{proof}

The following completeness theorem of the sequents with multiple definitions in the antecedent is an immediate consequence of Lemma~\ref{lem:comofmulti} and the reduction tree for sequents.

\begin{theorem}[Completeness for multiple definitions in the antecedent] \label{theo:completenessformultidefinitions}
Let $\Gamma$ and $\Delta$ be 
sets of PC-formulas and $D_1, \ldots, D_n$ definitions. If $\models D_1, \ldots, D_n, \Gamma \la \Delta$, then $D_1, \ldots, D_n, \Gamma \la \Delta$ is provable in \ps.
\end{theorem}

Then we have the following main completeness theorem.
\begin{theorem}[Completeness]\label{th:comp}
If $\models \Gamma \la \Delta$ and all definitions occurring either negatively in $\Gamma$ or positively in $\Delta$ are total, then $\Gamma \la \Delta$ is provable in \ps.
\end{theorem}

\begin{proof}
Let $\Gamma \la \Delta$ be a valid sequent such that any definition which occurs either negatively in $\Gamma$ or positively in $\Delta$ is total and let $T_S$ be a reduction tree with root $\Gamma \la \Delta$. Then by \ref{countleaf} of Proposition~\ref{prop:reductiontree}, all leaves of $T_S$ are valid. Since all leaves are of the form $D_1, \ldots, D_n, \Pi \la \Sigma$ where $\Pi$ and $\Sigma$ are
sets of atoms and $D_1, \ldots, D_n$ are definitions, it follows from Theorem~\ref{theo:completenessformultidefinitions} that they are provable in \ps. Hence, by \ref{prleaf} of Proposition~\ref{prop:reductiontree}, $\Gamma \la \Delta$ is provable in \ps.
\end{proof}

\section{Complexity results} \label{sec:complexity}
In this section, we provide some complexity results for PC(ID), which may give some helpful insight into the reasoning problems in PC(ID).

\begin{proposition} \label{prop:SAT(ID)}
Satisfiability problem in PC(ID) is NP-complete.
\end{proposition}

\begin{proof}
(Membership) Propositional well-founded models can be computed in polynomial time, e.g. using the algorithm of Van Gelder in~\cite{VanGelder91}. It is easy to define an algorithm that uses this well-founded semantics algorithm and finds models that satisfy PC(ID) theories in polynomial time on a non-deterministic turing machine.

(Hardness) Any satisfiability problem for propositional logic is trivially also a satisfiability problem for PC(ID).
\end{proof}

Recall Definition~\ref{def:totality} of totality of a definition $D$ with respect to a theory $T$: for each $I \models \bigwedge T$, the well-founded model of $D$ extending $\res{I}{\openp{D}}$ must be two-valued. Deciding totality is an interesting problem, not least because we cannot even formulate an inference rule to prove totality of a propositional inductive definition in the context of a certain set of PC(ID)-formulas.

\begin{proposition} \label{prop:totalitycomplexity}
Deciding whether a given propositional inductive definition is total with respect to a given propositional theory is co-NP-complete problem.
\end{proposition}

\begin{proof}
(Membership) Let $D$ be a propositional inductive definition, $T$ a propositional theory. Any interpretation $I$ such that $I \models \bigwedge T$ and the well-founded model of $D$ extending $\res{I}{\openp{D}}$ is not two-valued, is a certificate for the non-totality of $D$ with respect to $T$. Both checking whether $I \models \bigwedge T$ and whether the well-founded model of $D$ extending $\res{I}{\openp{D}}$ is two-valued can be done in polynomial time.

(Hardness) Consider the definition $D = \defin{ P \rul \neg P \land T}$. $D$ is total with respect to the empty theory if and only if $T$ is unsatisfiable. Thus we have found an instance of our decision problem that is equivalent to a co-NP-hard decision problem, namely unsatisfiability problem for propositional logic.
\end{proof}

\section{Conclusions, related and further work}

We presented a deductive system for the propositional fragment of FO(ID) which extends the sequent calculus for propositional logic. The main technical results are the soundness and completeness theorems of \ps. We also provide some complexity results for PC(ID).

Related work is provided by Hagiya and Sakurai in \cite{Hagiya84}. They proposed to interpret a (stratified) logic program as iterated inductive definitions of Martin-L\"of \cite{MartinLoef71} and developed a proof theory which is sound with respect to the perfect model, and hence, the well-founded semantics of logic programming. A formal proof system based on tableau methods for analyzing computation for Answer Set Programming (ASP) was given as well by Gebser and Schaub~\cite{Gebser06}. As shown in ~\cite{jelia/MarienGD04}, ASP is closely related to FO(ID). The approach presented in~\cite{Gebser06} furnishes declarative and fine-grained instruments for characterizing operations
as well as strategies of ASP-solvers and provides a uniform proof-theoretic framework for analyzing
and comparing different algorithms, which is the first of its kind for ASP.

The first topic for future work, as  mentioned in Section \ref{sec:intro}, is the development and implementation of a proof checker for \minisatid. This would require more study on resolution-based inference rules since \minisatid is basically an adaption of the DPLL-algorithm for SAT~\cite{DavisP60,DavisLL62}.

On the theoretical level, we plan to develop proof systems and decidable fragments of FO(ID). As mentioned in Section~\ref{sec:intro}, FO(ID) is not even semi-decidable and thus, a sound and complete proof system for FO(ID) does not exist. Therefore, we hope to build useful proof systems for FO(ID) that can solve a broad class of problems and investigate subclasses of FO(ID) for which they are decidable.

\section{Funding}
This work was supported by FWO-Vlaanderen and GOA/2003/08.
Johan Wittocx is a Research Assistant of the Fund for Scientific Research-Flanders (Belgium) (FWO-Vlaanderen).

\bibliographystyle{plain}
\bibliography{krr1}
\end{document}